\documentclass[runningheads]{llncs}
\usepackage{multirow}
\usepackage{expl3}
\usepackage[
	n,
	operators,
	advantage,
	sets,
	adversary,
	landau,
	probability,
	notions,	
	logic,
	ff,
	mm,
	primitives,
	events,
	complexity,
	asymptotics,
	keys]{cryptocode}

\usepackage{footnote}
\usepackage{csquotes}
\usepackage{todonotes}
\usepackage{caption}

\usepackage{graphicx}
\usepackage{hyperref}

\hypersetup{backref=true,       
                    pagebackref=true,               
                    hyperindex=true,                
                    colorlinks=true,                
                    breaklinks=true,                
                    urlcolor= black,                
                    bookmarks=true,                 
                    bookmarksopen=false,
                    linkcolor=red,
                    filecolor=black,
                    citecolor=blue,
}

\begin{document}
\title{Secure Sketch for All Sources (Noisy)}
%
\author{Yen-Lung Lai}%
%
%
%
\institute{Monash University Malaysia,\\ Jalan Lagoon Selatan, Bandar Sunway, 47500 Subang Jaya, Selangor\\
\email{\{yenlung.lai\}@monash.edu}}

%
\maketitle              
\begin{abstract}
Secure sketch produces public information of its input $w$ without revealing it, yet, allows the exact recovery of $w$ given another value $w'$ that is close to $w$. Therefore, it can be used to reliably reproduce any error-prone secret (i.e., biometrics) stored in secret storage. However, some sources have lower entropy compared to the error itself, formally called ``more error than entropy", a standard secure sketch cannot show its security promise perfectly to this kind of sources. This paper focuses on secure sketch. We propose an explicit construction for secure sketch. We show correctness and security to all sources with meaningful min-entropy at least a single bit. Besides, our construction comes with efficient recovery algorithm operates in polynomial time in the sketch size, which can tolerate high number of error rate arbitrary close to 1/2 for random error. The above result offers polynomial time solution to two NP-complete coding problems, suggesting P=NP.

\end{abstract}
\section{Introduction} \label{1.0}
Traditional cryptography systems rely on uniformly distributed and recoverable random strings for secret. For example, random passwords, tokens, and keys. These secrets must present exactly on every query for a user to be authenticated and get accessed into the system. Besides, it must also consist of high enough entropy, thus making it very long and complicated, further resulted in the difficulty in memorizing it. 
On the other hand, there existed plentiful non-uniform strings to be utilized for secrets in practice.  For instance, biometrics (i.e., human iris, fingerprint) which can be used for human recognition/identification purpose. Similarly, long passphrase (S. N. Porter, 1982 \cite{porter1982password}), answering several questions for secure access (Niklas Frykholm \textit{et al}., 2001~\cite{frykholm2001error}) or personal entropy system (Ellison \textit{et al}., 2000~\cite{ellison2000protecting}), and list of favorite movies (Juels and Sudan, 2006~\cite{juels2006fuzzy}), all are non-uniformly distributed random strings that can be utilized for secrets.

The availability of non-uniform information prompted the generation of uniform random string from non-uniform materials. Started by Bennett \textit{et al}., (1988)~\cite{bennett1988privacy}, identified two major steps in deriving a uniform string from noisy non-uniform sources. The first one is \textit{information-reconciliation}, by tolerating the errors in the sources without leaking any information. The second one refers to the \textit{privacy amplification}, which converts high entropy input into a uniformly random input. The information-reconciliation process can be classified into interactive (includes multi messages) and non-interactive (only includes single message) versions. For non-interactive line of work, it has been first defined by Dodis \textit{et al}., (2004)~\cite{dodis2004fuzzy} called the fuzzy extractor. Likewise, the fuzzy extractor used two steps to accomplish the task, which is the secure sketch (for error tolerance), and randomness extractor (for uniform string generation). Secure sketch is demanding because it enables information-reconciliation, e.g., exact recovery of a noisy secret while offering security assurance to it. Moreover, a secure sketch can be easily extended to fuzzy extractor for uniform string generation by using a randomness extractor. The generated random string can be used in independent security system for access control, identification, digital signature, etc.

This work focuses on secure sketch. We reviewed the limitations of current secure sketch constructions in Section \ref{1.1}. To overcome such limitations, we adopted the usage of resilience vector (RV) in Section \ref{5.0} to support better understanding of the structure of the sources. We proposed an explicit construction with included RV for sketching and recovery (secure sketch) in Section \ref{6.0} and \ref{7.0} respectively. Our proposed recovery mechanism has shown to be efficient in polynomial in the sketch size and allows error tolerance of error rate arbitrary close to $1/2$ (Section \ref{9.0}). In the end, in Section \ref{11.0} we formalize the security of our construction and show security to all sources with meaningful entropy at least single bit. We also compared our proposal with existing secure sketch construction, showing our construction enjoys the better lower bound of min-entropy requirement for a standard secure sketch (Section \ref{12.0}).

\subsection{Issues in Existing Secure Sketch Construction}\label{1.1}
Various secure sketch constructions can be found in the literature. Some notable constructions involved the code-offset construction proposed by Juels and Wattenberg (1999)~\cite{juels1999fuzzy} that operates perfectly over hamming matric space. Besides, Juels and Sudan (2006)~\cite{juels2006fuzzy} have also proposed another construction for metric other than hamming called the fuzzy vault. Besides, Dodis \textit{et al}., (2004)~\cite{dodis2004fuzzy} have proposed an improved version of the fuzzy vault, and also the Pin-sketch construction that relies upon generic syndrome encoding/decoding with $t$-error correcting BCH code $\cdv$, which works well for non-fixed length input over a universe $\mathcal{U}$. 

However, the above mentioned secure sketch construction only works for limited sources.  Briefly, given a point (some value) $w$, the sketch would allow the acceptance of its nearby point $w'$ within distance $t$ for exact recovery of $w$. Therefore, if an adversary can predict an accepting $w'$ with noticeable probability, the sketch must reveal $w$ to the adversary with noticeable probability as well. The tension between security and error tolerance capability is very strong. Precisely, the security is measured in terms of the residual (min-) entropy, which is the starting entropy of $w$ minus the entropy loss. Given some non-uniform sources with low min-entropy, especially, when the sources consist of \textit{more error than entropy} itself, deducting the entropy loss from the sources' min-entropy always outputs a negative value, hence, show no security. Because of this, correcting $t$ errors regardless of the structure of the input distribution would have to assume sufficient high min-entropy to the input sources. To show meaningful security for standard secure sketch, the min-entropy must at least half of the input length itself~\cite{dodis2009non}, hence, limiting the availability of secure sketch construction for low entropy sources.

Through exploitation of the struction of the input distributions, Fuller \textit{et al.}, (2013)~\cite{fuller2013computational} have show that the crude entropy loss over `more error than entropy' sources can be avoided by the measurement of {fuzzy min-entropy}, which defined as the min-entropy with maximized chances for a variable of $W$ within distance $t$ of $w'$:
\begin{align*}
\text{H}_{t,\infty}^\text{fuzz}{(W)}\overset{\text{def}}{=}-\log\left(\max\limits_{w'}\prob{W\in{B_{t}}(w')}\right),
\end{align*}
where $B_t(w')$ denoted a hamming ball of radius $t$ around $w'$. Conceivably, the fuzzy min-entropy is equivalent to the residual entropy, which is at least the min-entropy $\minentropy{W}$ minus the loss signified by the hamming ball $B_t(w')$ of radius $t$, s.t.
\begin{align*}
{\text{H}_{t,\infty}^\text{fuzz}{(W)}}\geq{\minentropy{W}-\log(B_t(w'))}.
\end{align*}
${\text{H}_{t,\infty}^\text{fuzz}{(W)}}$ is useful for security measurement instead of $\minentropy{W}$ especially when the residual entropy shows negative value (i.e. more error than entropy). However, due to the fact that ${\text{H}_{t,\infty}^\text{fuzz}{(W)}}$ depends upon the error tolerance distance $t$, and it is not necessary referring to the worst-case distribution for $W$, therefore, traditional way of showing security with ${\text{H}_{t,\infty}^\text{fuzz}{(W)}}$ measurement have to deal with such \textit{distribution uncertainty} by considering a family of distributions $\mathcal{W}$ for different variables i.e., $\lbrace{W_1,W_2,\ldots}\rbrace\in{\mathcal{W}}$ rather than single distribution. Viewed this way, ${\text{H}_{t,\infty}^\text{fuzz}{(W)}}$ measurement is only sufficient for computational secure sketch construction \cite{fuller2013computational}, \cite{woodage2017new}, which means that the security property of such construction only hold for computationally bounded attacker (i.e., polynomial time bounded) accompanies with strong assumption on the user has a precise knowledge over $\mathcal{W}$. However, it is unrealistic to assume every sources distribution can be modelled precisely, especially for high entropy sources like human biometric. 

\section{Overview Results}\label{2.0}
We highlighted our main four results as follow.

It was believed that the exploitation of the input structure is necessary \cite{fuller2016fuzzy} to construct a secure sketch for all sources. Follow in this way, our works adopted the principle of \textit{Locality Sensitive Hashing (LSH)} to generate a resilient vectors pair (trivially, a pair of longer strings with resilience property) for sketching and recovery. Details discussion on the resilient vector (RV) is covered in Section \ref{5.0}. The RV pair possessing resilience property, i.e., distance preserving that is useful for the exploitation of the input sources structure. Our first result is the metric of correlation measure between the RV pair and their input pair (Eq. \ref{eq:1} and Eq. \ref{eq:2}). 

Since the RV is used for sketching, such correlation measurement implies the entropy loss from the input. Therefore, the minimum entropy loss from the sketch reduced to the maximum correlation measured in between the RV pair, conditioned on their inputs. We formalize such minimum entropy loss based on the maximum tolerance distance $t_{\max}$ over any random input distribution, which implies maximum probability in looking for a nearby (similar) point within distance $t_{\max}$ (Corollary \ref{corollary:1}). This result is later being used for our reduction from fuzzy minimum entropy to Shannon entropy, to show necessary and sufficient condition for a system's security (see Section \ref{12.2}). 

Thirdly, we show that the minimum entropy loss of our construction could be at least three bits with BCH error correction codes. A Tighter result is also obtained by considering random error correction codes are used instead of BCH codes (Proposition \ref{proposition:3}), revealing the minimum entropy loss could be at least one bits (Eq. \ref{eq:15}). This pushed the lower bound of minimum entropy requirement of our secure sketch construction to accept any sources of entropy at least one bit, which is much lower compared to existing constructions. Nonetheless, above result is computational. It later is being used to derive the well-known information-theoretical bound (Shannon bound) and shown to be coincided with another computational bound commonly studied in coding theory research (Gilbert-Varshamov bound)  (see Section \ref{12.1}), lead us to the claim of computational secure sketch implies information-theoretical sketch. 

The last result we would like to highlight is the efficiency of the recovery algorithm in our construction. Without the consideration of the computational power in running the recovery algorithm, the recovery of the input from the sketch can be done with high probability (close to one) given the sketch size is large enough (Proposition \ref{proposition:1}). On the other hand, considering the computational power in running the recovery algorithm, we noticed that higher computational power, i.e., exponential time in the input (Eq. \ref{eq:9}) is needed in order to tolerate more errors. Nonetheless, such exponential computation time can be formalized to polynomial time in the sketch size to ensure efficient recovery while allowing more errors to be tolerated (Proposition \ref{proposition:2}). This result shows deep connection in between a difficult decoding problem over smaller metric space (of size $k^*$) could be reduced down to a more manageable problem over larger metric space $n>k^*$, suggesting P=NP.

\section{Preliminaries}\label{3.0}

There are some preliminaries to introduce the background of a standard secure sketch, entropy, and error correction code.\\

\noindent\textbf{Min-Entropy:} For security, one is always interested in the probability for an adversary to predict a random value, i.e., guessing a secret. For a random variable $W$, $\max\limits_{w}\prob{W=w}$ is the adversary's best strategy to guess the most likely value, also known as the predictability of $W$. The min-entropy thus defined as 
\begin{align*}
\minentropy{W}=-\log⁡(\max\limits_w\prob{W=w})
\end{align*}
min-entropy also viewed as worst-case entropy.\\

\noindent\textbf{Conditioned min-entropy:} Given pair of random variable $W$, and $W'$ (possible correlated), given an adversary find out the value $w'$ of $W'$, the predictability of $W$ is now become $\max\limits_{w}\condprob{W=w}{W'=w'}$. The conditioned min-entropy of $W$ given $W'$ is defined as 
\begin{align*}
\condminentropy{W}{W'}=-\log\left( \expsub{w'\leftarrow{W'}}{\max\limits_{w}\condprob{W=w}{W'=w'}}\right)
\end{align*}

\noindent\textbf{Error correction code:}~\cite{guruswami2004list} Let $q\geq{2}$ be an integer, let $[q]=\lbrace{1,\ldots,q}\rbrace$, we called an $[n,k,d]_q$-ary code $\cdv$ consist of following properties: 
\begin{itemize}
\item $\cdv$ is a subset of $[q]^n$, where $n$ is an integer referring to the \textit{blocklength} of $\cdv$.
\item The \textit{dimension} of code $\cdv$ can be represented as $\abs{\cdv}=[q]^k=V$
\item The \textit{rate} of code $\cdv$ to be the normalized quantity $\frac{k}{n}$
\item The \textit{min-distance} between different codewords defined as $\min\limits_{c,c^*\in{\cdv}}{\mathsf{dis}(c,c^*)}$
\end{itemize}

It is convenient to view code $\cdv$ as a function $\cdv:[q]^k\rightarrow{{[q]}^n}$. Viewed this way, the elements of $V$ can be considered as a message $v\in{V}$ and the process to generate its associated codeword $\cdv{(v)}=c$ is called \textit{encoding}. Viewed this way, encoding a message $v$ of size $k$, always adding redundancy to produce codeword $c\in{[q]^n}$ of longer size $n$. Nevertheless, for any codeword $c$ with at most $t=\lfloor{\frac{d-1}{2}}\rfloor$ symbols are being modified to form $c'$, it is possible to uniquely recover $c$ from $c'$ by using certain function $\mathsf{f}$ s.t. $\mathsf{f}(c')=c$. The procedure to find the unique $c\in{\cdv}$ that satisfied $\mathsf{dis}(c,c')\leq{t}$ by using $\mathsf{f}$ is called as \textit{decoding}. A code $\cdv$ is said to be efficient if there exists a polynomial time algorithm for encoding and decoding.\\

\noindent\textbf{Linear error correction code}~\cite{guruswami2004list}: Linear error correction code is a linear subspace of $\FF_q^n$. A $q$-ary linear code of blocklength $n$, dimension $k$ and minimum distance $d$ is represented as $\lbrack{n,k,d}\rbrack_q$ code $\cdv$.  For a linear code, a string with all zeros $0^n$ is always a codeword.  It can be specified into one of two equivalent ways with a generator matrix $G\in\FF_q^{n\times{k}}$ or parity check matrix $H\in\FF_q^{(n-k)\times{n}}$:

\begin{itemize}
\item a $\lbrack{n,k,d}\rbrack_q$ linear code $\cdv$ can be specified as the set $\lbrace{Gv:v\in{\FF_q^k}}\rbrace$ for an $n\times{k}$ metric which known as the \textit{generator matrix} of $\cdv$.
\item a $\lbrack{n,k,d}\rbrack_q$ linear code $\cdv$ can also be specified as the subspace $\lbrace{x:x\in{\FF_q^n}}$ and ${Hx=0^n}\rbrace$ for an $(n-k)\times{n}$ metric which known as the \textit{parity check matrix} of $\cdv$.
\end{itemize}

For any linear code, the linear combination of any codewords is also considered as a codeword over $\FF_q^n$. Often, the encoding of any message $v\in{\FF_q^k}$ can be done with $O(nk)$ operations (by multiplying it with the generator matrix, i.e., $Gv$. The distance between two linear codewords refers to the number of disagree elements between them, also known as the \textit{hamming distance}. 

Sometime, we refer $[n,k,d]$ code $\cdv$ as $[n,k,t]$ code $\cdv$ if the error tolerance distance $t$ is of interested rather than its minimum distance $d$.\\

\noindent\textbf{Secure sketch:} \cite{dodis2004fuzzy} An $(\mdv,m,\tilde{m},t)$-secure sketch is a pair of randomized procedures ``sketch'' ($\mathsf{SS}$) and ``Recover'' ($\mathsf{Rec}$), with the following properties:
\begin{itemize}
  \item[] $\mathsf{SS}$: takes input $w\in\mdv$ returns a secure sketch (e.g., helper string) $ss\in\bin^*$.
  \item[] $\mathsf{Rec}$: takes an element $w'\in\mdv$ and $ss$. If $\mathsf{dis}(w,w')\leq{t}$ for some tolerance threshold $t$, then $\mathsf{Rec}(w',ss)=w$ with probability $1-\beta$, where $\beta$ is some negligible quantity. If $\mathsf{dis}(w,w')>{t}$, then no guarantee is provided about the output of $\mathsf{Rec}$.
\end{itemize}

The security property of secure sketch guarantees that for any distribution $W$ over $\mdv$ with min-entropy $m$, the values of $W$ can be recovered by the adversary who observes $ss$ with probability no greater than $2^{-\tilde{m}}$. That is the residual entropy $\condminentropy{W}{W'}\geq\tilde{m}$.

\section{Main Idea}\label{4.0}
We here highlight some common notation to be used in this work, and a brief overview of our construction, focus on binary metric space.

\subsection{Notations} \label{4.1}
Let $\mdv_1=\bin^{k^*}$, and $\mdv_2=\bin^n$ denote two different sizes of metric spaces where $n>k^*$. The distance between different binary string $w$ and $w'$ denoted as $\mathsf{dis}(w,w')$ is the binary hamming distance (e.g., the number of disagree elements), i.e., ${\mathsf{dis}}(w,w')=\norm{w\xor{w'}}$ where $\norm{.}$ is the hamming weight that count the number of non-zero elements, and $\xor$ is the addition  modulo two operation (XOR). Besides, the error rate in between the input $w$ $w'\in{\mdv_1}$ is denoted as $\norm{w\xor{w'}}(k^*)^{-1}$ which is simply their normalized hamming distance. 

Our construction used two linear codes. We called one of these as `inner' code $\cdv_{in}$, and another one called `outer' code $\cdv_{out}$. Although not necessary, it is convenient to refer both chosen linear codes to be BCH codes \cite{berlekamp2015algebraic} with parameter $\lbrack{n^*,k^*,t^*}\rbrack_2$ for $\cdv_{in}$ and $\lbrack{n,k,t}\rbrack_2$ for $\cdv_{out}$, where $k^*\leq{n^*}<{k}\leq{n}$ holds. This is because the family of BCH codes has been well-studied and formed a large class of error correction code perfectly suit for binary input string. Moreover, it accompanies with efficient decoding algorithm $\mathsf{f}$ namely the syndrome decoding operates in $O(n^t)$\cite{berlekamp2015algebraic}. The tolerance rate of code $\cdv_{in}$ and $\cdv_{out}$ are denoted as $\xi^*={t^*}/{n^*}$ and  $\xi={t}/{n}$ respectively.

\subsection{Overview Construction}\label{4.2}
We here provide a brief overview of our sketching and recovery proposal for one to conceal a random string (encoded codeword) $c^*\in\bin^{n^*}$ while allows exact recovery of $c^*$ by using another noisy string $w'_{e'}\in{\bin^{k^*}}$ that is close to the noisy string $w_e$ used on sketching. 

\textbf{Sketching:} The proposed sketching procedure can be viewed as a two steps encoding process. First, one encodes $w$ using the `inner' code $\cdv_{in}$ to output a codeword $c^*$. Then, $c^*$ is randomly distorted by using some noisy string $w_e$, where the corrupted codeword, viewed as a syndrome vector $v_{syn}\in{\lbrace{0,1}\rbrace^{n^*}}$. The syndrome vector is then being padded with zeros to form a longer bit string and encoded by the `outer' code $\cdv_{out}$ to output the final codeword $c\in{\cdv_{out}}$. The final sketch is generated by concealing $c$ with an offset $\delta$ that is characterized by a pair of resilient vectors $\phi,\phi' \in{\bin^n}$. In different to existing construction, our sketching proposal consists of additional randomization with random noisy string $w_e$ used to distort the original codeword $c^*$. Then, the distorted codeword is being encoded into larger codeword $c$ and distorted again using RV to generate the final sketch. The noisy string introduced during sketching phase facilitates our studies of any random events potentially corrupting the codewords follows the random distribution of the noisy string $w_e$.

\textbf{Recovery:} For recovery, we are more interested in the event when the person in recovering $w$ does not know the distribution of the introduced noisy string $w_e$ during sketching phase. Therefore our recovery procedure is defined to be adversarial. Nevertheless, he/she can use some random noisy string $w'_{e'}$ and its corresponding RV, $\phi'$ to try to decode the corrupted codeword from the sketch. The zero-padded in front of the syndrome vector acts as an indicator to notify him/her the successfulness of the decoding results (i.e., first few bits are all zeros). More formally, our recovery procedure requires one to determine the distribution of noisy string $w_e$ introduced during sketching, and look for the similar noisy string $w'_{e'}$ (that viewed as the pre-image of $w_e$) i.e., $\norm{w'_{e'}\xor{w_e}}\leq{t'}$ to achieve error tolerance hence recover $w$ successfully.

\section{Resilient Vector: Properties and Generation} \label{5.0}

Since RV is a core element of our construction, we here provide details discussion on its properties and how it can be generated. The usage of RV into cryptography is first introduced by Ronald L. Rivest \cite{rivest2016symmetric} in 2016. Its main concept is derived from Locality Sensitive Hashing (LSH) defined as below.\\

\noindent\textup{\textbf{Locality Sensitive Hashing}~\cite{charikar2002similarity}} Given that $P_2>P_1$, while $w,w'\in{\mdv}$, and $\mathcal{H}={h_i:\mdv\rightarrow}U$, where $U$ refers to the output metric space (after hashing), which comes along with a similarity function $S$, where $i$ is the number of hash functions $h_i$. A locality sensitive hashing can be viewed as a probability distribution over a family $\mathcal{H}$ of hash functions follows $P_{h\in{\mathcal{H}}}\lbrack{h(w)=h(w')}\rbrack=S(w,w')$. In particular, the similarity function $S$ described the hashed collision probability in between $w$ and $w'$. 
\begin{align*}
P_{h\in{\mathcal{H}}}(h_i(w)=h_i(w'))\leq{P_1},  \ \ {\text{if}}\ S(w,w')<R_1\\
P_{h\in{\mathcal{H}}}(h_i(w)=h_i(w'))\geq{P_2},  \ \ {\text{if}}\ S(w,w')>R_2
\end{align*}
LSH transforms input $w$ and $w'$ to its output metric space $U$ with property that ensuring similarity inputs render higher probability of collision over $U$, and vice versa.

For RV generation, we only focus on a particular LSH family called hamming-hash~\cite{gionis1999similarity}. The hamming hash is considered as one of the easiest ways to construct an LSH family by bit sampling technique.  \\

\noindent\textbf{Hamming hash strategy}:
\textit{Let $[{k^*}]=\lbrace{1,\ldots,{k^*}}\rbrace$. For Alice with $w\in\bin^{k^*}$ and Bob with $w'\in\bin^{k^*}$. Alice and Bob agreed on this strategy as follow:}
\begin{enumerate}
	\item \textit{They are told to each other a common random integer $N\in{[{k^*}]}$.}
	\item \textit{They separately output \text{`0'} or \text{`1'} depend upon their private string $w$ and $w'$, i.e., Alice output \text{`1'} if the $N$-th bit of $w$ is \text{`1'}, else output \text{`0'}.}
	\item \textit{They win if they got the same output, i.e., $w(N)=w'(N)$.}
\end{enumerate}
Based on above strategy, we are interested in the probability for Alice and Bob outputting the same value. This probability can be described by a similarity function $S(w,w')=P$ where $P\in{\lbrack{0,1}\rbrack}$.

\begin{theorem}
\label{theorem:1}
Hamming hash strategy is a LSH with similarity function $S(w,w')=1-{\norm{w\xor{w'}}}({k^*})^{-1}$ 
\end{theorem}

The hamming hash strategy can also be operated in between Alice and Bob in an non-interactive way. To do so, Alice and Bob simply repeat Step 1 and Step 2 for $n$ number of times with a set of pre-shared integers $N=[N(1),N(2),\ldots,N(n)]\in{[k^*]^n}$ chosen randomly and uniformly over ${[k^*]^n}$. In the end, they can output a $n$ bits string $\phi$, and $\phi'$ respectively over ${\bin^n}$, which we have earlier named as \textit{resilient vectors}. We denote such non-interactive hamming-hash algorithm as $\mathsf{\Omega}:\mdv_1\times{[{k^*}]^n}\rightarrow{\mdv_2}$, which serves to sample the input binary string of size $k^*$ into a longer binary string a.k.a resilient vector of size $n>{k^*}$. 

Given input $w\in\bin^{k^*}$, and $N\in{[k^*]^n}$, algorithm $\mathsf{\Omega}:\mdv_1\times{[{k^*}]^n}\rightarrow{\mdv_2}$ can be described as follow:
\begin{center}
\fbox{%
\procedure[linenumbering]{$\mathsf{\Omega}(w,N)$}{%
  \phi\leftarrow \emptyset  \\
   \pcfor i=1,\ldots,n\pcdo\\
   \pcind\pcparse x=w(N(i))\pccomment{$x$ is the $N(i)$-th bits of $w$}\\
   \pcind\phi=\phi\concat{x}\\
   \pcendfor \\
   \pcreturn \phi
}
}
\end{center}

\begin{theorem}\label{theorem:2}
Suppose two resilient vectors $\phi, \phi'\in\bin^n$ are generated from $w,w'\in{\bin^{k^*}}$ respectively using hamming hash algorithm $\mathsf{\Omega}$ with a random integer string $N\in{\lbrack{{k^*}}\rbrack}^n$, then $\expect{\norm{\phi\xor\phi'}}=n\norm{w\xor{w'}}({k^*})^{-1}$.
\end{theorem}

\subsection{Correlation Measure in RVs}\label{5.1}
Let $\Phi$ and $\Phi'$ be two random variables over $\bin^n$, and $W$ and $W'$ be two random variables over $\bin^{k^*}$. Given a resilience vector $\phi\in{\Phi}$ generated from $w\in{W}$ with random string $N$, it follows $\Phi$ must correlate with ${W}$ where the probability to look for any random variable $\Phi\in{B_t(\phi')}$ (also means similar resilience vector s.t. $\norm{\phi\xor{\phi'}}\leq{t}$) varies conditioned on either ${W\not\in{B_{t'}(w')}}$ or ${W\in{B_{t'}(w')}}$. Note that ${W\in{B_{t'}(w')}}$ implies the inputs $w\in{W}$ and $w'\in{W'}$ must similar within distance $t'$ (e.g. $\norm{w\xor{w'}}\leq{t'}$), while ${W\in{B_{t'}(w')}}$ means $\norm{w\xor{w'}}>{t'}$. Such correlation can be measured by using the conditional probability described as ${\condprob{\Phi\in{B_t(\phi')}}{W\not\in{B_{t'}(w')}}}$ or ${\condprob{\Phi\not\in{B_t(\phi')}}{W\in{B_{t'}(w')}}}$ respectively. In particular, we are more interested in the maximum correlation, which can be expressed by the conditioned \textit{maximum} probability in looking for $\Phi\in{B_{t_{\max}}(\phi')}$ given ${W\not\in{B_{t'_{(+)}}(w')}}$ for some maximum distances $t_{\max}\geq{t}$ (over $\bin^{n}$) and $t'_{(+)}>t'$ (over $\bin^{k^*}$) defined as:
\begin{align}
&{\max\limits_{t=t_{\max}}{\condprob{\Phi\not\in{B_t(\phi')}}{{W\not\in{B_{t'_{(+)}}(w')}}}}}.\nonumber\\
&\geq\expsub{w'\leftarrow{W'}}{\max\limits_{\phi'}\condprob{\Phi\in{B_t(\phi')}}{W\not\in{B_{t'}(w')}}}.\label{eq:1}
\end{align} 
On the other hand, the conditioned \textit{maximum} probability in looking for $\Phi\not\in{B_t(\phi')}$ given ${W\in{B_{t'}(w')}}$ for some minimum distances $t_{\min}\leq{t}$ (over $\bin^{n}$) and $t'_{(-)}<t'$ (over $\bin^{k^*}$) is defined as:
\begin{align}
&{\max\limits_{t=t_{\min}}{\condprob{\Phi\not\in{B_t(\phi')}}{{W\not\in{B_{t'_{(-)}}(w')}}}}}\nonumber\\
&\geq\expsub{w'\leftarrow{W'}}{\max\limits_{\phi'}\condprob{\Phi\in{B_t(\phi')}}{W\not\in{B_{t'}(w')}}}\label{eq:2}
\end{align} 
\section{Sketching} \label{6.0}
We denote the sketching algorithm that employs the hamming-hash algorithm, $\mathsf{\Omega}$, an $\lbrack{n^*,{k^*},t^*}\rbrack_2$ `inner' code $\cdv_{in}$ and an $\lbrack{n,k,t}\rbrack_2$ `outer' code $\cdv_{out}$ as $\mathsf{SS}_{\mathsf{\Omega},\cdv_{in},\cdv_{out}}$. The sketching algorithm $\mathsf{SS}_{\mathsf{\Omega},\cdv_{in},\cdv_{out}}$ with random inputs $w, N$, and $\epsilon_{ss}$ is described as follow:
\begin{center}
\fbox{%
\procedure[linenumbering]{$\mathsf{SS}_{\mathsf{\Omega},\cdv_{in},\cdv_{out}}(w,N,\epsilon_{ss})$}{
   \mathcal{E}_{ss} \sample \lbrace{0,1}\rbrace^{{k^*}}\pccomment{initiate a random distribution $\mathcal{E}_{ss}$ with error parameter $\epsilon_{ss}$}   \\
      e \sample \mathcal{E}_{ss}\pccomment{sample $e$ uniformly at random from $\mathcal{E}_{ss}$, where $\norm{e}={{\floor{{{k^*}}\epsilon_{ss}}}}$}   \\
      c^*=\cdv_{in}{(w)}; \pccomment{encode $w$}\\
         w_e=w\xor{e}; \\
       v_{syn}={c^*}\xor{{(0^{n^*-k^*}\concat{w_e})}}\\
        v^*=0^{k-n^*}\concat{v_{syn}};\\ 
   c=\cdv_{out}{(v^*)}; \pccomment{encode $v^*$}\\
   \phi\leftarrow\mathsf{\Omega}(w_e,N)\\
   ss=c\xor{\phi};\\
   \pcreturn ss
}
}
\end{center}

Our sketching procedure consists of mainly two encoding stages. Given an random input string $w\in{\bin^{k^*}}$, the first encoding stage used $\cdv_{in}$ to encode a random string $w$ to generate a codeword $c^*\in{\bin^{n^*}}$. Viewed this way, $c^*$ can be any random codeword over $\cdv_{in}$, including the trivial codeword of all zeros i.e. $c^*=0^{n^*}$. Then, we generate a noisy string $w_e$ and pad it with zeros in front to generate a longer bit string, which can be viewed as the syndrome vector denoted as $v_{syn}={c^*}\xor{{(0^{n^*-k^*}\concat{w_e})}}$. The syndrome vector itself is also a codeword $v_{syn}\in\cdv_{in}$. Clearly, $v_{syn}$ conceals $c^*$ by using $w_e$, where the vector ${(0^{n^*-k^*}\concat{w_e})}$ acts as some random errors used to distort the original codeword $c^*$, lead to at most $\norm{{(0^{n^*-k^*}\concat{w_e})}}=\norm{w_{e}}$ bits flipped in $c^*$. Then, the second encoding stage used $\cdv_{out}$ to encode $v^*=0^{k-n^*}\concat{v_{syn}}$ to generate the final code word $c$. The $0^{k-n^*}$ zeros in front is used to notify the recovery algorithm if the decoding is success. The final sketch is formed by hiding $c$ with RV generated from $w_e$. 

For the realization of the noisy string $w_e$, we parse additional error to the original input $w$ using a random error vector $e\in\mathcal{E}_{ss}$ sampled uniformly at random follows some random distribution $\mathcal{E}_{ss}$. Such error distribution is parametrized by an error parameter $\epsilon_{ss}$ chosen within range $\epsilon_{ss}\in{[(2k^*)^{-1},1/4]}$. To be specific, all error vector $e\in{{\mathcal{E}_{ss}}}$ is of weight $\norm{e}={{\floor{{{k^*}}\epsilon_{ss}}}}$, and the generation of the noisy string follows $w_e=w\xor{e}$. The error vector $e$ is leaving in clear after it is parsed into the input $w$ to form $w_e$.

All steps on $\mathsf{SS}_{\mathsf{\Omega},\cdv_{in},\cdv_{out}}(w,N,\epsilon_{ss})$ can be done in $O({n^3})$, and the size of $ss$ is now depend upon the blocklength $n$ of the chosen `outer' code $\cdv_{out}$. 

Remark here the distribution for $w$ and $e$ are considered to be random in our case. The random string $N\in{[k^*]^{n}}$ is chosen in random and uniform over $[k^*]^{n}$. By LSH definition and Theorem \ref{theorem:2}, the generated RV is i.i.d, where every single bit of the RV follows the distribution of the noisy string $w_e$. Since the sketch $ss$ is generated by concealing $c$ using RV $\phi$, its distribution shall follows the generated RV, depends upon the input noisy string $w_e\in\bin^{k^*}$. In such a case, the same random string $N$ can be made public and reused, where the distribution of the generated sketch is only input dependence, i.e., depends upon the distribution of the noisy string $w_e$.

\section{Recovery}\label{7.0}
We denote the recover algorithm that employed the hamming-hash algorithm, $\mathsf{\Omega}$, an  $\lbrack{n^*,{k^*},t^*}\rbrack_2$ `inner' code $\cdv_{in}$ and an $\lbrack{n,k,t}\rbrack_2$ `outer' code $\cdv_{out}$ as $\mathsf{Rec}_{\mathsf{\Omega},\cdv_{in},\cdv_{out},\mathsf{f}}$. The recover algorithm $\mathsf{Rec}_{\mathsf{\Omega},\cdv_{in},\cdv_{out},\mathsf{f}}$ with inputs $ss$, $w'$, $N$, $\epsilon_{rec}$ to recover $w$ is described as follow:

\begin{center}
\fbox{%
\procedure[linenumbering]{$\mathsf{Rec}_{\mathsf{\Omega},\cdv_{in},\cdv_{out},\mathsf{f}}(ss, w',N,\epsilon_{rec})$}{%
\mathcal{E}_{rec}\sample\bin^{k^*} \pccomment  {initiate $\mathcal{E}_{rec}$ with error parameter $\epsilon_{rec}$}\\
\pcind\pcfor {i=1,\ldots,{|\mathsf{supp}(\mathcal{E}_{rec})|}}\\
\pcind e'_i\sample{\mathcal{E}_{rec}}\pccomment {sample $e'_i$ differently at random, where $\norm{e'_i}={{{\floor{{k^*}\epsilon_{rec}}}}}$}\\
\pcind w'_{e'_i}=w'\xor{e'_i}\\
 \pcind  \phi'_i\leftarrow\mathsf{\Omega}(w'_{e'_i},N)  \\
 \pcind   c'_i=ss\xor{\phi'_i} \pccomment{also $ss\xor{\phi'_i}=c\xor{(\phi\xor{\phi'_i})}$}\\
\pcind   c\leftarrow{\mathsf{f}(c'_i)}\pccomment{first decoding}\\ 
\pcind   \textbf{set} \ v^*=\cdv_{out}^{-1}(c) \\
\pcind\pcind\pcif{v^*[1],\ldots,v^*[k-n^*]=0^{k-n^*}}   \pccomment{first $k-n^*$ bits of $v^*$ are zeros}\\
  \pcind\pcind\pcind\textbf{set}\ v'_{syn} \text{ as the last $n^*$ elements of $v'^*$}\\
   \pcind\pcind\pcind   c'^*=v'_{syn}\xor{(0^{n^*-k^*}\concat{w'_{e'_i}})}\\
  \pcind\pcind\pcind   c^*\leftarrow{\mathsf{f}(c'^*)}\pccomment{second decoding}\\
 \pcind \pcind\pcind\pcreturn w=\cdv_{in}{^{-1}}(c^*)\\
 \pcind \pcind\pcind\textbf{break}\\
 \pcind \pcind\pcendif\\
 \pcind\pcendfor}}
\end{center}

Our recovery process is defined to be adversarial s.t. the person in recovering $w$ does not know the input error distribution $\mathcal{E}_{ss}$ introduced on sketching. Hence, different notation is used to denote the error distribution used in the recovery phase, e.g., $\mathcal{E}_{ss}$ for sketching, and $\mathcal{E}_{rec}$ for recovery. In particular, the non-trivial case (noisy) where the input $w$ is distorted by $e\in{\mathcal{E}_{ss}}$ is of our interest. Therefore it is good to set $\epsilon_{rec}\geq{\epsilon_{ss}}\geq(2k^*)^{-1}$ for the mean of determining the error vector $e\in{\mathcal{E}_{ss}}$ (introduced during sketching) for recovery of $w$. 

The recovery algorithm consists of mainly two decoding stages. The first decoding stage is designed to be an iterative decoding uses $\cdv_{out}$ and $\mathsf{f}$. It can be conveniently viewed as a brute-force decoding procedure of ${{|\mathsf{supp}(\mathcal{E}_{rec})|}}$ trials with recovery error distribution ${\mathcal{E}_{rec}}$ (parametrized by $\epsilon_{rec}$), means to recover the syndrome vector $v_{syn}$. For each iteration in the first decoding stage, if the first $k-n^*$ bits of $v^*$ are all zeros, the decoding is viewed as success and thus the recovery algorithm could proceed to the second decoding stage to recover $c^{*}$ and so $w$ from $v_{syn}$ using $\cdv_{in}$.  

The second decoding stage uses $\cdv_{in}$ and $\mathsf{f}$ to decode the corrupted syndrome vector (viewed as the corrupted codeword $c'^*$) to recover $c^{*}$. The decoding itself \textit{must} success if $\norm{w_e\xor{w'_{e'_i}}}\leq{t^*}$, thus $w$ can be recovered from $c^*$.

\section{Distribution Hiding with Random Error Parsing}\label{8.0}
In this section, we provide a discussion over the action of random error parsing is necessary to show security for a given source. In fact, the random error parsing process can also be interpreted as any randomization and perturbation process (i.e., hashing or encryption) applied to the input $w\in{\mdv_1}$ with some random error distribution $\mathcal{E}_{ss}$. Viewed this way, Step 4 of $\mathsf{SS}_{\mathsf{\Omega},\cdv_{in},\cdv_{out}}$ can be succinctly replaced by some random function $f:\mdv_1\times\mdv_{1}\rightarrow\mdv_1$ with input $w$ and $e$, to output $w_e$ as a random noisy string. Same function $f$ also applied on Step 4 of $\mathsf{Rec}_{\mathsf{\Omega},\cdv_{in},\cdv_{out},\mathsf{f}}$ with input $w'$ and different $e'_i$ to output $w_{e'_i}$. Nonetheless, since our focus only on error tolerance, we would stick to the simplest case where the random error parsing process is simply an addition modulo two (XOR) process.

We start from the sketching algorithm, which accepts any random input $w$ in some random distribution $W$ over $\mdv_1$. Such input would yield a random RV $\phi\in{\Phi}$ in random distribution $\Phi$ over $\mdv_2$ via RV generation using some public known random string $N$. Tolerating $t$ errors using an $[n,k,t]$ error correction code eventually reveal $W$. This is because the encoding process must ensure all random variable $W\in{B_{t}(w')}$ can be tolerated using decoding function $\mathsf{f}$, therefore, for any input string $w\in{W}$, encoding $w$ must store additional information that is manifested by the neighbourhood of $w$ within distance $t$. Doing so inevitably introduced $t$ information loss. When $t$ is larger than the entropy of the source, the source is said to loss all entropy and no security to show. 

To resolve the above issue, a straightforward way is to hide $W$ before adding redundancy to it. This can be done by parsing an error randomly and uniformly chosen from a list $\lbrace{e_1,\ldots,e_{|\mathsf{supp}(\mathcal{E}_{ss})|}}\rbrace\in{\mathcal{E}_{ss}}$ into the input string $w$ during the sketching phase. Doing so means the generated noisy string $w_e$ is now randomly and uniformly distributed over a family of distributions $\mathcal{W}$, i.e., $w_e\in{\mathcal{W}}=\lbrace{W_1,\ldots,W_{{|\mathsf{supp}(\mathcal{E}_{ss})|}}}\rbrace$. Because the distribution of RV is input dependence, it follows the generated RV $\phi\in{\Psi}=\lbrace{\Phi_1,\ldots,\Phi_{{|\mathsf{supp}(\mathcal{E}_{ss})|}}}\rbrace$ is also randomly and uniformly distributed in a family of i.i.d distribution $\Psi$. Remark here when $|\Psi|=|\mathcal{W}|=|\mathsf{supp}(\mathcal{E}_{ss})|=1$, it corresponds to the trivial case when the error vector $e\in{\mathcal{E}}_{ss}$ is all zeros (noiseless). This also means that the original distribution of the variable $W$ is now hidden over $\mathcal{W}$ randomly and uniformly. Therefore, for the non-trivial case, it is more appropriate to consider the family of distributions in $\mathcal{W}$ and $\Psi$ rather than single distribution in deriving the security of the sketch. Compared to the generic BCH encoding procedure, our proposed sketching algorithm can be viewed as a more general encoding procedure by considering a more general case where the error parameter $\epsilon_{ss}>{0}$ (or $\norm{e}\geq{0}$) is introduced during the sketching phase.

Because the sketch is generated by concealing the final codeword $c$ with an RV $\phi$. It follows that the worst-case security of the sketch is manifested by the maximum error tolerance distance over $\Psi$, where the probability to find a random variable $\Phi\in{B_{t_{\max}}(\phi')}$, i.e., similar RV (given $\phi'$) over $\Psi$ is maximum. In particular, we let $\xi=t/n={\lVert{w_e\xor{w'}}\rVert}({k^*})^{-1}$. For $\xi\geq{\epsilon_{ss}}$, such maximum distance (over $\bin^{n}$) can be described as $$t_{\max}=n(\xi-\epsilon_{ss})\geq{t}\geq{0}.$$
Recall the generated RV is input dependence (depends upon the noisy string $w_e$), thus, $\Phi$ is conditioned on the original distribution of the variable $W$. To look for $t_{\max}$, we first define a maximum tolerance distance (over $\bin^{k^*}$) as $$t'_{(+)}={(\xi+\epsilon_{ss}){k^*}},$$ for all distribution of $W\in{\mathcal{W}}$ after random error parsing, which is distributed (hidden) randomly and uniformly over $\mathcal{W}$. Remark here $t$ is now denoted to be input dependence. The original distribution of the variable $W$ may be different for different input string $w\in{W}$ before random error parsing. This means the distance ${\lVert{w_e\xor{w'}}\rVert}({k^*})^{-1}$ shall be different and so for the value of $t$ (or $\xi$) as well. Nevertheless, for $\epsilon_{ss}\geq{(2k^*)^{-1}}>0$, the defined maximum distance $t'_{(+)}$ would imply error rate of $\xi+\epsilon_{ss}$, which is always larger than $\xi$. Therefore, $t'_{(+)}$ is always maximum for any original distribution of the variable $W$, hold for all input string $w\in{W}$. 

The Corollary below characterized the worst-case security of RVs. Such security is measured in terms of conditioned maximum probability in getting similar RV within a maximum tolerance distance $t_{\max}\geq{t}\geq{0}$ over $\Psi$ given a noisy string $w'\in{W'}$ over some random distribution $W'\in{\mdv_1}$.
\begin{corollary}\label{corollary:1}
Given any random variable $W\in{\mathcal{W}}$, and a random string $w'\in{W'}$. For all RV over a family of distributions ${\Psi}$, the conditioned maximum probability to look for any similar RV $\phi$ in any random distribution $\Phi\in{B_{t_{\max}}(\phi')}$ over $\Psi$ (i.e. $\norm{\phi\xor{\phi'}}=\norm{\delta}\leq{t_{\max}}$ for all $\phi\in{\Phi}$) when ${W}\not\in{B_{t'_{(+)}}(w')}$ is measured to be
\begin{align}
&=\expsub{w'\leftarrow{W'}}{\max\limits_{\phi'}\condprob{\Phi\in{B_{t}(\phi')}}{{W}\not\in{B_{t'_{(+)}}(w')}}}\nonumber\\
&\leq{\max\limits_{t=t_{\max}}{\condprob{\norm{\delta}\leq{t}}{\norm{w_e\xor{w'}}\geq{t'_{(+)}}}}}\nonumber\\
&=\condprob{\norm{\delta}\leq{n(\xi-\epsilon_{ss})}}{\norm{w_e\xor{w'}}\geq{t'_{(+)}}}\leq{\exp{(-2n{\epsilon_{ss}^2)}}}\label{eq:3}
\end{align} 
\end{corollary}
\begin{proof}
$W\not\in{B_{t'_{(+)}}(w')}$ means ${\lVert{w_e\xor{w'}}\rVert}({k^*})^{-1}\geq{\xi+\epsilon_{ss}}$. For $\xi={\lVert{w_e\xor{w'}}\rVert}({k^*})^{-1}$, and $t_{\max}=n(\xi-\epsilon_{ss})$. It follows that $n\xi\geq{t}+n\epsilon_{ss}$ can be yielded by multiplying both sides of the inequality with $n$, yielded $t_{\max}\geq{t}$. Then the probability for ${\norm{\delta}}\leq{t_{\max}}$ given ${\lVert{w_e\xor{w'}}\rVert}({k^*})^{-1}\geq{\xi+\epsilon_{ss}}$ can be computed by \textit{Hoeffding's inequality} follows the last line of Eq. \ref{eq:3}.
\end{proof}  
By Corollary \ref{corollary:1}, the conditioned min-entropy of RV over $\Psi$ measured to be
\begin{align}
&\condminentropy{\Psi}{\mathcal{W}}=-\log\bigg({\expsub{W\leftarrow{\mathcal{W}}}{\condprob{\Psi=\Phi}{\mathcal{W}=W}}\nonumber}\bigg)\nonumber\\
&\geq{-\log\bigg(\expsub{w'\leftarrow{W'}}{\max\limits_{\phi'}\condprob{\Phi\in{B_{t_{\max}}(\phi')}}{W\not\in{B_{t'_{(+)}}(w')}}}\bigg)}\nonumber\\
&\geq{\log(1/\exp{(-2n{\epsilon_{ss}^2)}})}\label{eq:4}
\end{align}

\section{Correctness with Regardless Computational Power}\label{9.0}

Formally, the correctness characterizes the \textit{success rate }of the recovery of the original codeword $c^*$ from a sketch generated by $\mathsf{SS}_{\mathsf{\Omega},\cdv_{in},\cdv_{out}}$. In this section, we study the correctness theoretically \textit{without} relying on any efficiency argument over the algorithm pair $\langle\mathsf{SS}_{\mathsf{\Omega},\cdv_{in},\cdv_{out}},\mathsf{Rec}_{{\mathsf{\Omega}},\cdv_{in},\cdv_{out},\mathsf{f}}\rangle$ itself. 

Since $\cdv_{in}(w)=c^*$, given $\cdv_{in}$, revealing the value of $w$ implies the knowledge on $c^*$. Therefore, it is desired to show that the probability of success in revealing $w$ using algorithm pair $\langle\mathsf{SS}_{\mathsf{\Omega},\cdv_{in},\cdv_{out}},\mathsf{Rec}_{{\mathsf{\Omega}},\cdv_{in},\cdv_{out},\mathsf{f}}\rangle$ is \textit{at least} (minimum) $1-\beta$ with some negligible probability $\beta>0$. It can be expressed in Eq. \ref{eq:5} below for all $\epsilon_{rec}\geq{(2k^*)^{-1}}$
\begin{align}
\prob{\mathsf{Rec}_{\mathsf{\Omega},\cdv_{in},\cdv_{out},\mathsf{f}}(\mathsf{SS}_{\mathsf{\Omega},\cdv_{in},\cdv_{out}}(w,N,\epsilon_{ss}), w',N,\epsilon_{rec})=w}={1-\beta}\label{eq:5}.
\end{align} 
To measure the minimum probability for Eq. \ref{eq:5}, the maximum value of $\beta$ have to be computed, which also referring to the maximum error in recovering $c^*$ from a sketch with $\mathsf{Rec}_{\mathsf{\Omega},\cdv_{in},\cdv_{out},\mathsf{f}}$. 

To maximize such error for maximum value of $\beta$, one shall use a minimum distance $${t_{\min}}=n(\xi+\epsilon_{ss})\leq{t}$$ for any random variable $\Phi\in{\Psi}$ in looking for a distinct RV (given $\phi'$), i.e., $\Phi\notin{B_{t_{\min}}(\phi')}$. 

Recall the distribution of RV is input dependence, hence, $\Phi$ is conditioned on the original distribution of the variable $W$. To look for $t_{\min}$, we first define a minimum tolerance distance as $$t'_{(-)}={(\xi-\epsilon_{ss}){k^*}}$$ for any distributions of the variable $W\in{\mathcal{W}}$ after random error parsing. Clearly, for $\epsilon_{ss}\geq{(2k^*)^{-1}}>0$, the defined minimum distance $t'_{(-)}$ would imply error rate of $\xi-\epsilon_{ss}$, which is always smaller than $\xi$. At same point, one always need $t'_{(-)}\leq{t^*}<t$ holds so that $t'_{(-)}$ is minimum for any original distribution of the variable $W$ (before random error parsing), which hold for all $w\in{W}$. The following Corollary revealing $\beta\leq{\exp(-2n\epsilon^2_{ss})}$.
\begin{corollary}\label{corollary:2}
Given any random variable ${W}\in{\mathcal{W}}$, and a random string $w'\in{W'}$. For all RV over a family of distributions ${\Psi}$, the conditioned maximum probability to look for any distinct RV $\phi$ in any random distribution $\Phi\not\in{B_{t_{\min}}(\phi')}$ over $\Psi$ (i.e. $\norm{\phi\xor{\phi'}}=\norm{\delta}\geq{t_{\min}}$ for all $\phi\in{\Phi}$) when $W\in{B_{t'_{(-)}}(w')}$ is measured to be
\begin{align}
&\expsub{w'\leftarrow{W'_0}}{\max\limits_{\phi'}\condprob{\Phi\not\in{B_{t_{\min}}(\phi')}}{W\in{B_{t'_{(-)}}(w')}}}\nonumber\\
&\leq{\max\limits_{t=t_{\min}}{\condprob{\norm{\delta}\geq{t}}{\norm{w_e\xor{w'}}\leq{t'_{(-)}}}}}\nonumber\\
&=\condprob{\norm{\delta}\geq{n(\xi+\epsilon_{ss})}}{\norm{w_e\xor{w'}}\leq{t'_{(-)}}}\leq{\exp{(-2n{\epsilon_{ss}^2)}}}.\label{eq:6}
\end{align} 
\end{corollary}
\begin{proof}
$W\in{B_{t'_{(-)}}(w')}$ means ${\lVert{w_e\xor{w'}}\rVert}({k^*})^{-1}\leq{\xi-\epsilon_{ss}}$. For $\xi={\lVert{w_e\xor{w'}}\rVert}({k^*})^{-1}$, and $t_{\min}=n(\xi+\epsilon_{ss})$, then $n\xi\leq{t}-n\epsilon_{ss}$ can be yielded by multiplying both sides of the inequality with $n$, yielded ${t_{\min}}\leq{t}$. The probability for ${\norm{\delta}}\geq{t_{\min}}$ given ${\lVert{w_e\xor{w'}}\rVert}({k^*})^{-1}\leq{\xi-\epsilon_{ss}}$ can be computed by \textit{Hoeffding's inequality} follows the last line of Eq. \ref{eq:6}. 
\end{proof}  
The following Proposition can be obtained by comparing Eq. \ref{eq:5} and Eq. \ref{eq:6}.
\begin{proposition}
\label{proposition:1}
For all $\epsilon_{ss}\in{[(2k^*)^{-1},1/4]}$, $\epsilon_{rec}\geq{(2k^*)^{-1}}$, and $\norm{w_e\xor{w'}}\leq{t'_{(-)}}\leq{t^*}$, 
\begin{align*}
&\prob{\mathsf{Rec}_{\mathsf{\Omega},\cdv_{in},\cdv_{out},\mathsf{f}}(\mathsf{SS}_{\mathsf{\Omega},\cdv_{in},\cdv_{out}}(w,N,\epsilon_{ss}), w',N,\epsilon_{rec})=w}\\
&\geq{{1-{\exp{(-2n{\epsilon_{ss}^2)}}}}}
\end{align*}
\end{proposition}
Above result showed the recovery of $c^*$ and so $w$ will success with high probability when $n$ is sufficiently large without considering the efficiency of the algorithm pair $\langle\mathsf{SS}_{\mathsf{\Omega},\cdv_{in},\cdv_{out}},\mathsf{Rec}_{{\mathsf{\Omega}},\cdv_{in},\cdv_{out},\mathsf{f}}\rangle$.

\section{Correctness with Regard to Computational Power}\label{10.0}
This section provided details discussion over the correctness of the algorithm pair $\langle\mathsf{SS}_{\mathsf{\Omega},\cdv_{in},\cdv_{out}},\mathsf{Rec}_{{\mathsf{\Omega}},\cdv_{in},\cdv_{out},\mathsf{f}}\rangle$ as well as their efficiency in term of computation complexity. Since $k^*<n$, we use $n$ as upper bound and define both $\langle\mathsf{SS}_{\mathsf{\Omega},\cdv_{in},\cdv_{out}},\mathsf{Rec}_{{\mathsf{\Omega}},\cdv_{in},\cdv_{out},\mathsf{f}}\rangle$ to be \textit{efficient} if they can run in polynomial time $\poly$ in the input sketch size $n$ to show correctness. Clearly, all steps on $\mathsf{SS}_{\mathsf{\Omega},\cdv_{in},\cdv_{out}}$ can be done in $\poly$, our efficiency arguments would therefore only focus on the recovery algorithm $\mathsf{Rec}_{{\mathsf{\Omega}},\cdv_{in},\cdv_{out},\mathsf{f}}$ itself.
\subsection{Computational Reduction of Correctness}\label{10.1}
Recall our correctness claim in Section \ref{9.0} comes with bounded probability of error such that the recovery would success with probability at least $1-\exp(-2n\epsilon_{ss}^2)$ given $\norm{\delta}\leq{t_{\max}}$ holds. In this subsection, we show that the above derived error bound can be reduced to merely depends upon  the algorithm itself, parametrized by the number of zero padding in front of $v_{syn}$.

To show this, recall $\mathsf{Rec}_{\mathsf{\Omega},\cdv_{in},\cdv_{out},\mathsf{f}}$ will only proceed to the second decoding if the first decoding return $v^*$ with the first $k-n^*$ bits are all zeros. Because the selection of the error vector $e'\in{\mathcal{E}_{rec}}$ is random, every iteration of the first decoding should return a random codeword $c\in{\cdv_{out}}$, hence its first $k-n^*$ bits of $v_{syn}$ are random over $\bin^{k-n^*}$. Such argument holds much stronger when random generator matrices are used for both $\cdv_{in}$ and $\cdv_{out}$ (see coming Subsection \ref{10.3}). In this regard, the probability for first $k-n^*$ bits of $v_{syn}$ are all zeros can be described as $2^{-{(k-n^*)}}$. Once the error vector $e'$ is found (or first decoding is viewed as succeed), it follows $\norm{w_e\xor{w'}}\leq{t'_{(-)}}\leq{t^*}$, where the second decoding \textit{must} success to recover $c^*$ with probability $1-2^{-{(k-n^*)}}$ with error revealed by the number of zeros padding to the syndrome vector $v_{syn}$. Nevertheless, a necessary condition for above argument to hold would be the case when
\begin{align}\label{eq:7}
\beta\leq{\exp(-2n\epsilon^2_{ss})}\leq{2^{-(k-n^*)}},
\end{align}
which yielded a tighter upper bound for the derived recovery error $\beta$. Without the exact knowledge of $\epsilon_{ss}$, there is no straightforward way to determine the value of $n$ which often yield a meaningful computational reduction to describe the maximum error exactly as $\beta=2^{-(k-n^*)}$. Nonetheless, since $\epsilon_{ss}\in{[(2k^*)^{-1},1/4]}$, the minimum value of $n$ can be set to $n={2{(k^*)}^2}$, it follows (for $k-n^*=1$) $$\beta\leq{\exp(-2n\epsilon_{ss}^2)}=\exp(-4k^2\epsilon_{ss}^2)={\exp(-1)}<2^{-(k-n^*)}=2^{-1}.$$ Doing so can ensure Eq. \ref{eq:7} holds for any $\epsilon_{ss}\in{[(2k^*)^{-1},1/4]}$ and the recovery error is maximum which can be described exactly as $\beta=2^{-(k-n^*)}=1/2$. To achieve lower $\beta$, higher value of $n\geq{2{(k^*)}^2}$ is needed. This suggested that reducing the recovery error $\beta$ always need sufficient large $n\geq{2{(k^*)}^2}$ and there is no straightforward way of deducing it given the value of $\epsilon_{ss}$ is unknown.

Above reasoning showed the dependency of our correctness derived previously in Section \ref{9.0} reduced to the number of zeros padding in the sketching phase, measured as $k-n^*$, conditioned on the value of $n$ is sufficient large. Noting that such reduction is computational. It means our derived correctness now holds for any chosen parameters for $\cdv_{in}$ and $\cdv_{out}$ with $\epsilon_{ss}\in{[(2k^*)^{-1},1/4]}$ as long as $n$ is sufficiently large where $k^*\leq{n^*<k\leq{n}}$ follows. 

%

For instance, under the designation of an BCH code \cite{peterson1972error}, its correctness is defined using some positive integer $m'\geq{3}$. Given a desired value of tolerance distance $t<2^{m'-1}$, one can construct an $\lbrack{n,k,t}\rbrack$ BCH code $\cdv_{out}$ with parameters $n=2^{m'}-1$, $n-k\leq{m't}$ and minimum distance $d\geq{2t-1}$. Then, for sufficiently large $n$, it is succinct to express
$$-\log(1/\beta)={k-n^*}=m'\geq{3}.$$ 
Doing so means that once the first decoding is success, its error is bounded by $\beta\leq{0.125}$ corresponds to the number of zeros padding in front of $v_{syn}$, i.e.,  at least three zeros. Eventually, the second decoding stage shall success with  probability at least $1-\beta\geq{0.875}$ where $\norm{w_e\xor{w'}}\leq{t'_{(-)}}\leq{t^*}$ follows. 

\subsection{Tolerating More Error Computationally}\label{10.2}
So far, our derived correctness theoretically demonstrating at most $t'_{(-)}\leq{t^*}$ of errors (or error rate of at most $\xi-\epsilon_{ss}$) can be tolerated over any random distribution $W\in\mathcal{W}$. It is natural to ask whether one can tolerate more error (more than $t‘_{(-)}$) and what is the maximum achievable error tolerance rate of our construction.

In fact, it is showed in Corollary \ref{corollary:1} it is infeasible to tolerate more than $t'_{(+)}$ errors because the probability to look for similar RVs is exponentially small if $\norm{w_e\xor{w'}}\geq{t'_{(+)}}$. However, it remains possible for one to tolerate $t_{(+)}=\floor{t'_{(+)}}\leq{t'_{(+)}}$ number of errors, significantly more than some minimum  value $t_{(-)}=\floor{{t'_{(-)}}}\leq{{t'_{(-)}}}$ computationally. 

To show this, for any error (distance) described as ${\lVert{w_e\xor{w'}}\rVert}={t_{(+)}}$. By using an error parameter $\epsilon_{rec}\in{[\epsilon_{ss},2\epsilon_{ss}]}$, the \textit{final error} can be described as ${\lVert{w_e\xor{w'_{e'}}}\rVert}={t_{(+)}}\pm\floor{k^*\epsilon_{rec}}$. In view of this, given high enough value of $\epsilon_{rec}$, i.e. $\epsilon_{rec}=2\epsilon_{ss}$ is used for recovery, any error of ${\lVert{w_{e}\xor{w'}}\rVert}=t_{(+)}$ is possible to be reduced down to ${\lVert{w_e\xor{w'_{e'}}}\rVert}=t_{(-)}$. Eventually, the remaining errors $\norm{w_e\xor{w'_{e'}}}={t_{(-)}}\leq{t'_{(-)}}$ can be tolerated with high probability follows Proposition \ref{proposition:1}. 

To do so, suppose one does not know the value of $\epsilon_{ss}$, he/she can try to choose some value of $\epsilon_{rec}\in[(2k^*)^{-1},1/2]$ (since $\epsilon_{ss}\in[(2k^*)^{-1},1/4]$) during recovery. Doing so allow him/her to generate a list of possible error vectors $\lbrace{{e'_1},\ldots,{e'_{|\mathsf{supp}(\mathcal{E}_{rec})|}}}\rbrace\in{\mathcal{E}_{rec}}$, which corresponds to a list of possible noisy strings $w'_{e_i}\in{W'_i}$ (for $i=1,\ldots,|\mathsf{supp(\mathcal{E}_{rec})}|$) over another family of distributions $\mathcal{W'}=\lbrace{{W'_1},\ldots,W'_{{|\mathsf{supp}(\mathcal{E}_{rec})|}}}\rbrace$. For each chosen value of $\epsilon_{rec}$, one shall have $\mathsf{Rec}_{\mathsf{\Omega},\cdv_{in},\cdv_{out},\mathsf{f}}$ runs in ${{|\mathsf{supp}(\mathcal{E}_{rec})|}}$ iterations to try all possible $\lbrace{{e'_1},\ldots,{e'_{|\mathsf{supp}(\mathcal{E}_{rec})|}}}\rbrace\in{\mathcal{E}_{rec}}$ until he/she found a noisy string $w'_{e'}$ s.t. $\norm{w_e\xor{w'_{e'}}}={t_{(-)}}\leq{{t'_{(-)}}}$ holds. Note that our recovery process generally covered the trivial case as well when $\floor{k^*\epsilon_{rec}}=\floor{k^*\epsilon_{ss}}=0$ (noiseless case).

More precisely, let $d'=\floor{k^*\xi}$ denotes the original distance $\norm{w\xor{w'}}={d'}$. Parsing an error $e\in{\mathcal{E}_{ss}}$ of weight $\floor{k^*\epsilon_{ss}}$ to the input $w\in{W}$ yields $w_e$. It follows there should be a resultant worst-case error (described in terms of \textit{maximum} distance) where $$\norm{w_e\xor{w'}}={\floor{k^*\xi}}\pm\floor{k^*\epsilon_{ss}}={d'\pm\floor{k^*\epsilon_{ss}}}
={d'+\floor{k^*\epsilon_{ss}}}=t_{(+)}.$$ 
Note that when $\norm{w_e\xor{w'}}={d'-\floor{k^*\epsilon_{ss}}}={t_{(-)}}\leq{t'_{(-)}}$ (refers to `-' sign), such errors is trivial which can be tolerated with probability at least $1-2^{-(k-n^*)}$ with sufficiently large $n$. Our goal is to tolerate the worst-case error, described as $t_{(+)}={d'+\floor{k^*\epsilon_{ss}}}$ (refers to the `+' sign). For some random value of $\epsilon_{ss}$ introduced in sketching, it follows the worst-case error concerning the case when the value $\epsilon_{rec}$ is maximum, i.e. $\epsilon_{rec}={2\epsilon_{ss}}$. To be specific, there must have a solution for $e'$ in the worst-case (when $\epsilon_{rec}={2\epsilon_{ss}}$) described as: 
\begin{align}
&\norm{w_e\xor{w'_{e'}}}=\norm{(w\xor{e})\xor({w'\xor{e'}})}=\norm{(w\xor{w')}\xor{(e\xor{e'}})}\nonumber\\
&={(d'}+\floor{k^*\epsilon_{ss}})-\floor{2k^*\epsilon_{ss}}={d'-\floor{k^*\epsilon_{ss}}}={t_{(-)}}.\label{eq:8}
\end{align}
Clearly, ${t_{(-)}}=\floor{t'_{(-)}}\leq{t'_{(-)}}\leq{t^*}$, the error would therefore be tolerated with high probability given sufficiently large $n$.

Nevertheless, to achieve such goal without the knowledge of $\epsilon_{ss}$, one has to trial $e'\in{\mathcal{E}_{rec}}$ of different weights for Eq. \ref{eq:8} to hold. Hence, the recovery complexity is always bounded by the maximum brute-force trials with maximum tolerate distance $t_{(+)}$. In view of this, ones have to deal with the issue of computational power in running $\mathsf{Rec}_{\mathsf{\Omega},\cdv_{in},\cdv_{out},\mathsf{f}}$, where the maximum number of iteration $|\mathsf{supp}(\mathcal{E}_{rec})|$ has to be identified to determine maximum recovery complexity for efficiency investigation.


With $\epsilon_{rec}=2\epsilon_{ss}$, by \textit{Stirling approximation}, the value ${|\mathsf{supp}(\mathcal{E}_{rec})|}$ can be bounded as
\begin{align}
{{|\mathsf{supp}(\mathcal{E}_{rec})|}}={{{k^*}\choose{\floor{k^*\epsilon_{rec}}}}}={{{k^*}\choose{\floor{2k^*\epsilon_{ss}}}}}\leq{2^{{k^*{h_2(2\epsilon_{ss})}}}},\label{eq:9}
\end{align}
where $h_2(x)=-x\log(x)-(1-x)\log(1-x)$ is the binary entropy function with input error rate of $x$. 

Although it remains feasible for one to have the recovery algorithm run in parallel with all possible value of $\epsilon_{rec}\in[(2k^*)^{-1},1/2]$ using different computational machine. The overall brute-force complexity is still bounded by the order of  $O(2^{{k^*{h_2(2\epsilon_{ss})}}})$ for $\epsilon_{ss}\in{[(2k^*)^{-1},1/4]}$.

\subsection{Choosing the Best Code for Difference Sources}\label{10.3}

Recall that $d'=\floor{k^*\xi}$, where $\xi=\norm{w_e\xor{w'}}(k^*)^{-1}$. It is obvious that the value of $t_{(+)}$ defined in our construction strictly depends upon the input noisy string's distribution $w_e\in{\mathcal{W}}$ (or the value of $\xi$). In such a case, the maximum value for $\xi$ reveals the maximum achievable distance $t_{(+)}$ for worst-case error tolerance. Such achievable tolerance distance does not necessary restricted by a particular choice of error correction code, rather, it is input dependence, i.e. depends upon the value of $d'$.

Nonetheless, it is straightforward for one to simply choose a code with specified tolerance rate $\xi=t/n$ for $\cdv_{out}$. In such a case, the value of $d'=\floor{k^*\xi}$ can be determined exactly and so the achievable tolerance distance for worst-case error $t_{(+)}=\floor{2k^*\xi}$ is set immediately depends upon the chosen code for $\cdv_{out}$. However, doing so is not always a good practice in reality. This is because if the error rate introduced during sketching phase is larger than the error tolerance rate of the chosen code, i.e., $\epsilon_{ss}>\xi$, the \textit{correctness} of the recovery algorithm might not hold. On the other hand, if $\epsilon_{ss}<\xi$, it may lead to \textit{over error tolerance} (i.e., tolerating more error than necessary), hence cannot show security to\textit{ more error than entropy} sources. In view of this, the best one can do is to look for the minimum solution of $\xi=\epsilon_{ss}$ (recall $\xi\geq{\epsilon_{ss}}$). Note that doing so is equivalent to look for the minimum distance $d'$ for the input pair $(w,w')$ where the error rate $\epsilon_{ss}$ introduced during the sketching phase has to be identified. Once such minimum solution is found, it follows that the knowledge on $t_{(+)}$ is immediate.

Since $\xi\geq{\epsilon_{ss}}$, it follows $t_{(-)}\geq{0}$. Therefore, the required \textit{minimum distance} for $w$ and $w'$ in our construction can be simply described with $\epsilon_{ss}$ as
\begin{align}
\norm{w\xor{w'}}=d'=\floor{k^*\xi}\geq{{\floor{k^*\epsilon_{ss}}}}\geq{0}.\label{eq:10}
\end{align}
By minimum $d'=\floor{k^*\xi}=\floor{k^*\epsilon_{ss}}=0$, it follows 
\begin{align}\label{eq:11}
{t_{(-)}}=0\leq\norm{w_e\xor{w'_{e'}}}\leq{t^*}
\end{align}
must hold in order for the second decoding with $\cdv_{in}$ to  success. Hence, one can always choose the parameters $t^*\geq{0}$ for $\cdv_{in}$ which is trivial for any $[n^*,k^*,t^*]$ linear code. 

Same thing applied to the outer code $\cdv_{out}$. Given such minimum solution is found (means $\xi=\epsilon_{ss}$). One shall have $t_{\max}=n(\xi-
\epsilon_{ss})=0$. Therefore, one can simply choose $\cdv_{out}$ to be any trivial linear code with $t=t_{\max}=0$. For the most easiest way, simply choose two random generator matrices for $\cdv_{in}$ and $\cdv_{out}$ with $t^*=0$ and $t=0$ respectively. Doing so means no error tolerance thus $c^*={c^*}'$ and $c=c_i'$ are assumed. In such a case, one can skip both of the decoding steps (relies on $\mathsf{f}$) and replace it with Gaussian elimination which can be done efficiently in $O(n^3)$.

Focusing on the worst-case error when $\epsilon_{rec}=2\epsilon_{ss}$ (follows Eq. \ref{eq:8}), the first affirmative decoding result (first $k-n^*$ bits of ${v^*}$ are zeros) would mean
\begin{align*}
\norm{w_{e}\xor{w'_{e'}}}={(d'}+\floor{k^*\epsilon_{ss}})-\floor{k^*\epsilon_{rec}}={0}.
\end{align*}
Thus, one shall have the worst-case solution for $\epsilon_{rec}$, i.e., $\epsilon_{rec}=2\epsilon_{ss}$, coincides with the minimum solution when $\xi=\epsilon_{ss}$ s.t.
\begin{align}\label{eq:12}
\epsilon_{rec}=\epsilon_{ss}=\xi.
\end{align}
For sufficiently large $n$, the recovery of $c^*$ would success with probability $1-2^{-(k-n^*)}$ and the maximum tolerance distance $t_{(+)}=d'+\floor{k^*\epsilon_{ss}}=\floor{2k^*\epsilon_{ss}}$ can also be determined. Since $\epsilon_{ss}\in{[(2k^*)^{-1},1/4]}$, it means the solution $\xi\leq{1/4}$ follows. Thus, one able to tolerate at most 
\begin{align}\label{eq:13}
t_{(+)}=d'+\floor{k^*\epsilon_{ss}}=\floor{k^*\xi}+\floor{k^*\epsilon_{ss}}\leq\floor{2k^*\epsilon_{ss}}=\floor{k^*/2}
\end{align}
number of errors in the worst-case using our construction with maximum error rate $\epsilon_{ss}=1/4$ introduced in sketching phase. 

\subsection{Computational Hardness of Recovery: The NP-Complete problems}\label{10.4}

Nonetheless, Eq. \ref{eq:9} suggesting exponential computation time in the input size $k^*$ for all $\epsilon_{rec}\in[(2k^*)^{-1},1/2]$, which is highly inefficient for large $k^*$. This result is not suprised because determining the minimum distance of an error correction code or the \textit{minimum distance problem} is indeed NP-complete \cite{vardy1997intractability}. Besides, looking for such nontrivial error vector $e'\in{\mathcal{E}_{rec}}$, which viewed as the minimum-weight solution to Eq. \ref{eq:8} is another NP-complete problem \cite{mceliece1978inherent}, commonly refer to the \textit{maximum likelihood decoding problem} for a linear code  \cite{guruswami2005maximum}. Clearly, if the later one can be solved easily, then the former problem can also be solved easily by trying $\floor{k^*\epsilon_{rec}}=1,2\ldots$ (parametrized by $\epsilon_{rec}$) with different random string $w'$ to yield a polynomial length of yes/no solutions. More formally, the corresponding decision version of maximum likelihood decoding problem can be formalized as follow:
\\

\noindent\textbf{Problem}: \textit{Maximum likelihood decoding\\}
\textbf{Instance}: \textit{A $\alpha\times{n}$ binary matrix $H$, a vector $y\in\FF_{2}^{\alpha}$, and integer $z>0$\\}
\textbf{Question}: \textit{Is there a vector $x\in{\FF_{2}^{n}}$ of weight $\leq{z}$, such that $Hx=y$}\\
\\Clearly, the answer for above question would be `yes' if one able to find $x$, which can be viewed as the solution for $e'\in{\mathcal{E}_{rec}}$ where Eq. \ref{eq:8} holds. More explicitly, let $H$ be the parity check matrix of $\cdv_{out}$, so $\alpha=n-k$. Note that $Hc\Leftrightarrow{\mathsf{syn}}(c)=0^\alpha$, where ${\mathsf{syn}}(c)$ is the syndrome of the encoded codeword $c\in{\cdv_{out}}$. Given the corrupted codeword $c'$, the RV $\phi'$ and $H$, the syndrome decoding algorithm $\mathsf{f}$ computes 
\begin{align*}
\mathsf{f}(c')=H\cdot{(c\xor{(\phi'\xor\phi)}})=Hc\xor{H\delta}=0^\alpha\xor{\mathsf{syn}(\delta)}={\mathsf{syn}(\delta)}
\end{align*}
yields the syndrome of $\delta$. Given  the distance $\norm{\phi\xor\phi'}=\norm{\delta}$ is small, i.e., $\norm{\delta}\leq{t}$, knowing $\mathsf{syn}(\delta)$ is enough to determine the offset $\delta$, which can be done via $\mathsf{f}$ in $\poly$ running time. 

In our case, by treating $x=\delta$ follows the constrain of $\norm{\delta}\leq{t_{\max}}$ (see the proof in Corollary \ref{corollary:1}), the hardness of looking for $\delta$ (of size $n$) is as hard as the above mentioned decision problem related to the maximum likelihood decoding. Nevertheless, recall results in Section \ref{10.1} showed that if $e'\in{\mathcal{E}_{rec}}$ can be found s.t. $\norm{w_e\xor{w'_{e'}}}\leq{t_{(-)}}\leq{t^*}$ holds, then $\norm{\delta}\leq{t_{\max}}$ can be achieved with probability at least $1-2^{-(k-n^*)}$ (conditioned on sufficiently large $n$). Viewed this way, efficient searching process for $e'$ immediately implies efficient maximum likelihood decoding algorithm, hence resolve the NP-complete problem discussed above in an efficient manner.  

More precisely, the searching process for the error vector $e'$ can be formalized by the algorithm $\langle\mathsf{SS}_{\mathsf{\Omega},\cdv_{in},\cdv_{out}},\mathsf{Rec}_{{\mathsf{\Omega}},\cdv_{in},\cdv_{out},\mathsf{f}}\rangle$ to start and end within $\floor{k^*\epsilon_{rec}}=0,1,2,\ldots,\floor{2k^*\epsilon_{ss}}$ corresponds to the worst-case error rate $\epsilon_{ss}$ introduced during sketching. 
Because the recovery process is defined adversarial, the actual value of $\epsilon_{ss}$ might not known. Nonetheless, we have $\epsilon_{ss}\in{[(2k^*)^{-1},1/4]}$, hence, the maximum number of trial would be bounded at most $\floor{2k^*\epsilon_{ss}}=\floor{k^*/2}$ for a given value of $k^*$.


\subsection{Deterministic Polynomial Time Recovery}\label{10.5}
In the last subsection, we have studied that the successful recovery of the original codeword $c^*$ using our proposed recovery algorithm is computational hard in the sense that the process in looking for the nontrivial error vector $e'$ is NP-complete. This subsection provided more details discussion on the efficiency of our proposed recovery algorithm. More precisely, we would show that certain minimum requirement (bound) is required to be satisfied in order to claim efficiency of our proposed algorithm, hence resolved the above mentioned NP-complete problems.

Note that the efficiency of our proposed recovery algorithm $\mathsf{Rec}_{\mathsf{\Omega},\cdv_{in},\cdv_{out},\mathsf{f}}
$ strongly depends upon the brute-force complexity itself, which is proportional to the value of ${|\mathsf{supp}(\mathcal{E}_{rec})|}$. Given one any chosen values of $n^*$ and $k$ for $\cdv_{in}$ and $\cdv_{out}$ respectively s.t. ${{{k^*{h_2(2\epsilon_{ss})}}}}\leq{k-n^*}$, such complexity can be bounded in terms of the sketch size $n$ described as:
\begin{align}\label{eq:14}
|{{\mathsf{supp}(\mathcal{E}_{rec})|}}\leq{2^{{k^*{h_2(2\epsilon_{ss})}}}}\leq{2^{k-n^*}}={2^{m'}}=n+1.
\end{align}
Then, by summing all the necessary operation steps (complexity in term of $n$), i.e., $\text{Step 1}+\text{Step 2}+....+\text{Step 16}=1+n+1+...+1=\poly$ over the recovery algorithm $\mathsf{Rec}_{{\mathsf{\Omega}},\cdv_{in},\cdv_{out},\mathsf{f}}$. The overall complexity of the recovery algorithm can be expressed as some polynomial function $\poly$ using $n$. 


\subsection{Efficiency in Average-case and Worst-case scenarios}\label{10.6}{Our efficient result discussed in previous section suggested padding more zeros in front of $v_{syn}$ means greater value of $n$, hence higher recovery complexity. It is natural for one to question on whether expressing the number of iterations in term of $n+1$ is meaningful. More formally, one may ask 

\textit{``Does such expression really help in reducing our effort in solving the given problem?"} 
 
In this subsection, we will discuss the controversy mentioned above and show that such expression is necessary for large class error rate but not all, i.e., $\epsilon_{rec}<1/2$.  

To start the discussion, it is good to recall the prerequisite bound of our efficiency claim showed in Eq. \ref{eq:14}. Noting that the number of zeros padding $k-n^*=m'$ shall increase proportional to the value of $h_2(2\epsilon_{ss})$ for it to hold. This can be argued in the sense that given any recovery error parameter $\epsilon_{rec}\in{[(2k^*)^{-1},2\epsilon_{ss}]}$, one have to try $2^{k^*h_2(\epsilon_{rec})}$ iterations to look for the solution of error vector $e'\in{\mathcal{E}_{rec}}$ of weight $\floor{k^*\epsilon_{rec}}$. Since $\cdv_{in}(w)=c^*$, the total number of possible value for $c^*$ must be bounded by the possible values of $w\in{W}$, which is at most $2^{k^*}$. In other words, if there are total number of $2^{m'}$ possible values for $c^*$, at least $2^{m'}\leq{2^{k^*}}$ number of random guesses are needed to ensure one can always reveal $c^*$ exactly. 

Under the \textit{average-case scenario} when $\epsilon_{ss}\in{[(2k^*)^{-1},1/4)}$ (or equivalently $\epsilon_{rec}\in{[(2k^*)^{-1},1/2)}$), one shall have the recovery algorithm to try maximally $2^{k^*h_2(2\epsilon_{ss})}$ iterations to search for $e'$ in recovering $c^*$ from the worst-case error. Since the number of possible value for $c^*$ is denoted as $2^{m'}$, the searching for $e'$ shall stop at most $2^{k^*h_2(2\epsilon_{ss})}<{2^{k^*}}={2^{m'}}-1$ iterations to yield a meaningful efficiency claim in the sense that it can always do better than random guesses of at least $2^{m'}=2^{k-n^*}$ trials manifested by the number of zero paddings in front of $v_{syn}$. 


In fact, it is practically more efficient to express $2^{m'}=\poly+1\leq{2^{k^*}}$ rather than $n+1$. This is because despite Eq. \ref{eq:14} holds with linear size of complexity $O(n)$ (by expressing $2^{m'}=n+1$), large value of $n$ would lead to longer operation time (practically) for the remaining step of $\mathsf{Rec}_{\mathsf{\Omega},\cdv_{in},\cdv_{out},\mathsf{f}}
$, which are bounded in $O(n^3)$. Even if expressing $2^{m'}=\poly+1\leq{2^{k^*}}$ is suffice to show our recovery algorithm should operate in $\poly$ polynomial time in the sketch size.  Keeping the value $n$ small indeed offers better efficiency guaranty both in term of complexity and practice, since small $n$ implies \textit{both} $O(n^3)$ and $\poly$ are small. Therefore, meaningful efficiency claim implies at most $$2^{k^*h_2(2\epsilon_{ss})}=2^{m'}-1=2^{-(k-n^*)}-1=\poly\leq{2^{k^*}}-1$$ iteration in searching $e'$ over larger class of worst-case error rate which can be precisely described as 
\begin{align*}
&2\epsilon_{ss}\leq{h_2^{-1}\bigg(\frac{\log(2^{m'}-1)}{k^*}\bigg)}=h_2^{-1}\bigg(\frac{\log(\poly)}{k^*}\bigg)\\
&=h_2^{-1}\bigg(\frac{\log(2^{(k-n^*)}-1)}{k^*}\bigg)<h_2^{-1}\bigg(\frac{k-n^*}{k^*}\bigg)\leq{1/2}.
\end{align*}  
Then, follows Eq. \ref{eq:7}, the recovery error can be described exactly as $$\beta=2^{-(k-n^*)}=2^{-m'}=1/(\poly+1)\leq{1/\poly}.$$ 
Above results demonstrated that the recovery error is at most $1/\poly$, hence tighter upper bound can be used to describe $\beta=1/\poly$. This means that the recovery error $\beta$ is also $\poly$ reducible. It follows sufficiently large value of $n$ accompanies with correctness in reducing the error itself described as $\beta=1/\poly$. In such a case, it is succinct to describe $k-n^*$ using some polynomial function $\poly\leq{2^{k^*}}-1$ to show efficient error tolerance for large class of worst-case error, arbitrary close to fraction of $1/2$.

On the other hand, under the \textit{worst-case scenario} when $\epsilon_{ss}=1/4$ (or equivalent $\epsilon_{rec}=1/2$), one would need $2^{k^*h_2(1/2)}=2^{k^*}$ number of iterations in searcing for $e'$. Therefore, no any advantages can be gained by using our recovery algorithm to recover the exact value of $c^*$ compared to random guessing of $2^{k^*}$ trials for all possible values of $c^*$. Hence, meaningful efficiency claim cannot achieve in this case. Above scenario can be well described using the notion of \textit{perfect secrecy}, first studied by C. Shannon (1949) \cite{shannon1949communication}.}

Based on the above discussion, one can conclude that expressing the number of iterations in term of $k-n^*=\log(\poly+1)\leq{k^*}$ is always necessary for average-case scenario but not for worst-case scenario to achieve meaningful recovery efficiency claim with $\poly$ complexity. This result lead to the refined Eq. \ref{eq:14} and yield the requirement (bound) of meaningful efficiency claim over large class of worst-case error for $\mathsf{Rec}_{\mathsf{\Omega},\cdv_{in},\cdv_{out},\mathsf{f}}
$ described as \begin{align}\label{eq:15}
k^*h_2(\epsilon_{rec})\leq{{k-n^*}},
\end{align}
where $$k-n^*=\log(\poly+1)\leq{k^*},$$ holds for average-case scenario.
\section{Propositions Formalisation}\label{11.0} 
This section provide summary of our derived results of correctness and efficiency for $\langle\mathsf{SS}_{\mathsf{\Omega},\cdv_{in},\cdv_{out}},\mathsf{Rec}_{{\mathsf{\Omega}},\cdv_{in},\cdv_{out},\mathsf{f}}\rangle$. In particular, we formalized them into two major Propositions in the following paragraphs.

By Eq. \ref{eq:7} - \ref{eq:15}, we can formalise the following Proposition to characterize the correctness of $\mathsf{Rec}_{\mathsf{\Omega},\cdv_{in},\cdv_{out},\mathsf{f}}
$ for $t_{(+)}$ number of error in efficient manner with BCH codes.
\begin{proposition} \label{proposition:2}
Let $k^*\geq{k-n^*}\geq{3}$ where $k-n^*=\log(\poly+1)$. Given any error rate $\epsilon_{ss}\in{[(2{k^*})^{-1},1/4)}$, $\epsilon_{rec}\in{[(2k^*)^{-1},2\epsilon_{ss}]}$, and $\norm{w\xor{w'}}\geq{0}$ s.t. 
\begin{align*}
 {{{k^*h_2{(\epsilon_{rec})}}}}\leq{k-n^*}
\end{align*}
Then, for sufficiently large $n$, there exits BCH codes for $\cdv_{in}\in{\bin^{k^*}}$ and $\cdv_{out}\in{\bin^{n}}$ where $$\prob{\mathsf{Rec}_{\mathsf{\Omega},\cdv_{in},\cdv_{out},\mathsf{f}}(\mathsf{SS}_{\mathsf{\Omega},\cdv_{in},\cdv_{out}}(w,N,\epsilon_{ss}), w',N,\epsilon_{rec})=w}\geq{0.875}$$ can be achieved efficiently operating in time $\poly$ and the maximum achievable error tolerance rate is $2\epsilon_{ss}<{1/2}$.
\end{proposition}
It is natural to ask whether the bound for $k-n^*\geq{3}$ can be generalized to $k-n^*\geq{1}$, thus one obtains a tighter bound for $m'$ derived in Eq. \ref{eq:7}. Noting that a tighter bound for $m'\geq{k-n^*}\geq{1}$ implies the existence of larger class of $[n,k,t]$ linear codes which can be determined using our construction (by looking the most suitable value of $\xi$) for any sources $w$ (given $w'$) of minimum distance at least $\norm{w\xor{w'}}={\floor{k^*\xi}}\geq{\floor{k^*\epsilon_{ss}}}\geq{0}$. 

Nonetheless, we will show that the answer for this question is positive which allow one to simply choose two random generator matrices for $\cdv_{in}$ and $\cdv_{out}$ given some error parameter $\epsilon_{ss}\in{[(2k^*)^{-1},1/4)}$. In such a case, the selection of $\cdv_{in}$ and $\cdv_{out}$ no longer necessary to be BCH codes. Simply making sure the value of $n$ is sufficiently large to show correctness with proper number of zeros padding to form $v_{syn}$ and $v^*$, it follows that if $\exp(-2n\epsilon^2_{ss})\leq{1/2}$, then Eq. \ref{eq:7} should hold for $k-n^*\geq{1}$. Thus, $\prob{\mathsf{Rec}_{\mathsf{\Omega},\cdv_{in},\cdv_{out},\mathsf{f}}(\mathsf{SS}_{\mathsf{\Omega},\cdv_{in},\cdv_{out}}(w,N,\epsilon_{ss}), w',N,\epsilon_{rec})=w}\geq{0.5}$. One could also run the recovery with different value of $\epsilon_{rec}\in[(2k^*)^{-1},1/2)$ or try $\floor{k^*\epsilon_{rec}}=0,1,2,\ldots,\floor{k^*/2}$ until the first affirmative answer is obtained to determine the most suitable tolerable error $t_{(+)}$ as discussed in the last section. The recovery procedure can always done better than random guessing and its complexity can be bounded by some polynomial function $\poly$ follows the efficiency argument discussed in Section \ref{10.6}. Based on above reasoning, we shall have the following proposition to formalize $\langle\mathsf{SS}_{\mathsf{\Omega},\cdv_{in},\cdv_{out}},\mathsf{Rec}_{{\mathsf{\Omega}},\cdv_{in},\cdv_{out},\mathsf{f}}\rangle$ as an efficient sketching and recovery algorithm for all inputs of minimum distance $\norm{w\xor{w'}}\geq{0}$.

\begin{proposition} \label{proposition:3}
Let $k^*\geq{k-n^*}\geq{1}$ where $k-n^*=\log(\poly+1)$. Given $\epsilon_{ss}\in{[(2{k^*})^{-1},1/4)}$, $\epsilon_{rec}\in[(2k^*)^{-1},2\epsilon_{ss}]$, and $\norm{w\xor{w'}}\geq{0}$ s.t. 
\begin{align*}
&{{k^*h_2{(\epsilon_{rec})}}}\leq{{k-n^*}}
\end{align*}
Then, for sufficiently large $n$, there exits $[n,k,t]_2$ linear codes for $\cdv_{in}\in{\bin^{k^*}}$ and $\cdv_{out}\in{\bin^{n}}$ where $$\prob{\mathsf{Rec}_{\mathsf{\Omega},\cdv_{in},\cdv_{out},\mathsf{f}}(\mathsf{SS}_{\mathsf{\Omega},\cdv_{in},\cdv_{out}}(w,N,\epsilon_{ss}), w',N,\epsilon_{rec})=w}\geq{0.5}$$ can be achieved efficiently operating in time $\poly$ and the maximum achievable error tolerance rate is $2\epsilon_{ss}<{1/2}$.
\end{proposition}
The results in Proposition \ref{proposition:2} and \ref{proposition:3} have also suggested the NP-complete problems mentioned in \cite{vardy1997intractability} and \cite{mceliece1978inherent} can be solved efficiently in polynomial time by using algorithm pair $\langle\mathsf{SS}_{\mathsf{\Omega},\cdv_{in},\cdv_{out}},\mathsf{Rec}_{{\mathsf{\Omega}},\cdv_{in},\cdv_{out},\mathsf{f}}\rangle$. Follows the work in \cite{karp1972reducibility} stated, if any NP-complete problem possesses a polynomial time algorithm to solve, then so does every NP-problem, and hence we shall have P=NP.

\section{Security}\label{12.0}
Recall the minimum information required to show correctness using $\mathsf{Rec}_{\mathsf{\Omega},\cdv_{in},\cdv_{out},\mathsf{f}}
$ could be expressed as the number of zeros padding to $v_{syn}$, which is $k-n^*$ as long as $n$ is sufficiently large. Any adversary should be able to gain certain minimum information due to the zeros padding. This minimum information leakage is independent of how liberal or conservative is the selection of the input string $w$. Rather, it is a limit (lower bound), which required to be set in before constructing any error tolerance system, to satisfy certain minimum requirement of correctness. Formally, follows Eq. \ref{eq:4} and Eq. \ref{eq:7}, such minimum information leakage can be used to express the conditioned min-entropy of RV:
\begin{align*}
&\condminentropy{\Psi}{\mathcal{W}}\geq{\floor{\log(1/\exp{(-2n{\epsilon_{ss}^2)}})}}\geq{m'}={k-n^*}\geq{3}.
\end{align*}
Noting that the result in Proposition \ref{proposition:3} would give tighter result, leads to
\begin{align}
&\condminentropy{\Psi}{\mathcal{W}}\geq{\floor{\log(1/\exp{(-2n{\epsilon_{ss}^2)}})}}\geq{m'}={k-n^*}\geq{1}.\label{eq:16}
\end{align}
Above results showed a tight lower bound min-entropy requirement for any sources of at least bit $(m\geq{1})$. Note that such entropy measurement is computational in the sense that it merely depends upon how the user parametrize the algorithm pair $\langle\mathsf{SS}_{\mathsf{\Omega},\cdv_{in},\cdv_{out}},\mathsf{Rec}_{{\mathsf{\Omega}},\cdv_{in},\cdv_{out},\mathsf{f}}\rangle$ for $\cdv_{in}$ and $\cdv_{out}$ with $k^*\leq{n^*}<{k\leq{n}}$.

\subsection{Computational Secure Sketch Implies Information-Theoretic Secure Sketch}\label{12.1}

This section gives a more general discussion on the implication of a standard secure sketch is in fact computational derivable.

Fuller \textit{et al}., \cite{fuller2013computational} have showed negative result on computational relaxation of information-theoretical secure sketch requirement. In particular, for hamming metric space,  computational secure sketch assumed to retain high Hill entropy (pseudo-random entropy) implies error correction code with approximately $2^{m'}$ points that correct $t$ random error. Their proof relies on the existence of Shannon code, which can correct arbitrary random errors $t'\leq{t}$. The same result has first implicitly stated by Dodis \textit{et al.}, (see Section 8.2 in \cite{dodis2004fuzzy}) where the main idea is by the conversion of worst-case error into random error by using random permutation function $\pi$ which being part of the public sketch. Nevertheless, the proof itself is non-constructive because it was still an open question on how to construct an efficient Shannon code. 

It is not difficult to verify that our findings agreed on above argument with explicit construction of the algorithm pair $\langle\mathsf{SS}_{\mathsf{\Omega},\cdv_{in},\cdv_{out}},\mathsf{Rec}_{\mathsf{\Omega},\cdv_{in},\cdv_{out},\mathsf{f}}\rangle$ viewed as an efficient Shannon code. More precisely, the sketching algorithm $\mathsf{SS}_{\mathsf{\Omega},\cdv_{in},\cdv_{out}}$ serves to encode any noisy input $w$ and the recovery algorithm $\mathsf{Rec}_{\mathsf{\Omega},\cdv_{in},\cdv_{out},\mathsf{f}}$ can be used for decoding. 

We now show that our construction meets the well-studied computational bound (Gilbert-Varshamov (GV) Bound) and coincides with the info-theoretic bound (Shannon bound) as an optimal error correction code.

To do so, we shall use the derived computational bound for $\mathsf{Rec}_{\mathsf{\Omega},\cdv_{in},\cdv_{out},\mathsf{f}}
$ to be efficient and correct with probability at least $1-1/\poly$ is given as (see Proposition \ref{proposition:3} and Eq. \ref{eq:15}):
\begin{align}\label{eq:17}
{{{k^*h_2{(\epsilon_{rec})}}}}\leq{k-n^*} \  (computational).
\end{align}

%

Let $(k-n^*)/k^*=\log(\poly+1)/k^*=1-R$ where $R$ is of rate $0\leq{R}\leq{1-1/k^*}$ (since $\log(\poly+1)\leq{k^*}$ and $k-n^*\geq{1}$). Apply such result to Eq. \ref{eq:17}, we obtain: 
\begin{align}\label{eq:18}
R\leq{1-h_2(\epsilon_{rec})}\leq{1-h_2(\epsilon_{ss})}\leq{1-h_2(1/2k^*)}<{1} \ \textit{(Shannon Bound)},
\end{align}
which is indeed the Shannon bound with any source of size $k^* \geq{2}$ (where $k\geq{3}$) for it to be meaningful. Same thing applied for the upper bound which attained the GV bound for $R$ in average-case scenario (i.e., $\epsilon_{ss}\in{[(2k^*)^{-1},1/4)}$) described as: 
\begin{align}\label{eq:19}
R\geq{1-h_2({2\epsilon_{ss}})}>{0} \ (GV-Bound).
\end{align}
Note that when the number of zero padding $k-n^*>k^*$, the rate $R$ becomes negative. Therefore, restricting the number of zero padding $k-n^*\leq{k^*}$ is necessary to ensures the achievable rate $R$ is positive. Note that the value of $R$ is inversely proportional to the order of $\poly$ and the sketch size $n$. Given the sketch size $n$ is decreased, it follows $\poly$ is decreased thus the rate $R$ should be increased. 

Observe that both Eq. \ref{eq:18} (Shannon bound) and Eq. \ref{eq:19} (GV bound) can achieve rate $R<1$ arbitrary close to one given $k^*$ is sufficiently large. Applied this result into our construction, this implies for any input of size $k^*\geq{2}$, for sufficiently large $n$,  there exist a good encoding and decoding function which can be characterized by algorithm pair $\langle\mathsf{SS}_{\mathsf{\Omega},\cdv_{in},\cdv_{out}},\mathsf{Rec}_{\mathsf{\Omega},\cdv_{in},\cdv_{out},\mathsf{f}}\rangle$ where $k^*\leq{n^*}<{k}\leq{n}$ follows. Nonetheless, these two bound are subjected to the same minimum rate $R>0$, where same upper bound of source entropy loss can be described using $\log(\poly)\leq{\log(2^{k^*}}-1)$ for the average-case scenario and $\log(\poly+1)=\log(2^{k^*})=k^*$ for the worst-case scenarios.

Based on above reasoning, computational secure sketch implies information theoretic secure sketch, where both of them are subjected to the same upper bound and lower bound (where $R\in{(0,1)}$) with entropy loss at most $k^*$. 

Follows the correctness formalized in Proposition \ref{proposition:3}, and the security reasoning above, the Proposition below is given to characterize our construction as a efficient secure sketch for average-case scenario.
\begin{proposition} \label{proposition:4} For $m'\geq{1}$. Algorithm pair $\langle\mathsf{SS}_{\mathsf{\Omega},\cdv_{in},\cdv_{out}},\mathsf{Rec}_{{\mathsf{\Omega}},\cdv_{in},\cdv_{out},\mathsf{f}}\rangle$ is an efficient $(\mdv_1,m,m',t_{(+)})$-secure sketch with parameter $k^*\leq{n^*}<{k}\leq{n}$ (for sufficiently large $n$). It can correct fraction of worst-case errors arbitrary close to $1/2$ (i.e., $t_{(+)}<\floor{k^*/2}$).
\end{proposition}
\subsection{Reduction from Fuzzy min-entropy to Shannon entropy}{\label{12.2}}
As computational secure sketch implies information theoretic ones, it is natural to ask how our construction related to the usage of fuzzy min-entropy in defining the system security. Nonetheless, we here show that Shannon entropy is indeed a necessary and sufficient condition to show security for any sources with non-zero min-entropy. 

To support our claim that Shannon entropy is necessary and sufficient, we recall the statement saying that a list of possible RVs in different distributions $\lbrace{\Phi_1,\ldots,\Phi_{{|\mathsf{supp}(\mathcal{E}_{ss})|}}}\rbrace\in{\Psi}$ can be generated from a list of possible noisy strings over a family of distributions $\lbrace{W_1,\ldots,W_{{|\mathsf{supp}(\mathcal{E}_{ss})|}}}\rbrace\in{\mathcal{W}}$ after the random error parsing. Meanwhile, the original distribution of random variable $W$ where $w\in{W}$ is concealed in $\mathcal{W}$. Looking for $w$ is equivalent to look for $W$ over $\mathcal{W}$, which means meaningful security can be claimed by measuring the \textit{minimum} fuzzy min-entropy of every possible distribution for $W$ over $\mathcal{W}$ (after random error parsing). However, we argued that it is difficult to precisely model all possible distribution, especially for large entropy distributions. This is because large entropy sources usually accompany with large input size $k^*$ and large number of points $2^{m'}$. Modelling individual distribution of large number of points $2^{m'}$ require high computational cost and computationally inefficient. 

Arguing that the worst-case security should offer maximum probability in getting similar RVs within maximised tolerance $t_{\max}$ (see Section \ref{8.0}). We therefore characterize the maximum probability using the conditioned minimum entropy of the RVs (given the noisy strings). Such result can easily derived using Hoeffding bound described in Eq. \ref{eq:4}, we restate it here:
\begin{align*}
&\condminentropy{\Psi}{\mathcal{W}}=-\log\bigg({\expsub{W\leftarrow{\mathcal{W}}}{\condprob{\Psi=\Phi}{\mathcal{W}=W}}\nonumber}\bigg)\nonumber\\
&\geq{-\log\bigg(\expsub{w'\leftarrow{W'}}{\max\limits_{\phi'}\condprob{\Phi\in{B_{t_{\max}}(\phi')}}{W\not\in{B_{t'_{(+)}}(w')}}}\bigg)}\nonumber\\
&\geq{\log(1/\exp{(-2n{\epsilon_{ss}^2)}})}.
\end{align*}  
It is not difficult to see that above equation computes the \textit{minimum} fuzzy min-entropy among all possible RVs distributions (over $\Psi$) concerning to maximum tolerance distance $t_{\max}$ (conditioned on $t'_{(+)}$). Such result yielded a reduction from fuzzy min-entropy to minimum entropy measurement over the original distribution of the variable $W$ concerning $t'_{(+)}$, hold for all input string $w\in{W}$. 

However, merely showing reduction from fuzzy min-entropy to min-entropy is not enough, our goal is to reduce min-entropy to Shannon entropy. Such reduction is in fact done with the formalization of Proposition \ref{proposition:2} and \ref{proposition:3} with sufficiently $n$ and number of zeros padding $k-n^*$. These result eventually lead us to the derivation of the Shannon bound and GV bound in computational manner (Eq. \ref{eq:18} and \ref{eq:19}). Noting that the derived bounds on Eq. \ref{eq:18} and \ref{eq:19} have removed the $o(1)$ term compared to the actual bounds derived in \cite{varshamov1957estimate}
\cite{gilbert1952comparison} \cite{shannon1948mathematical}: 
\begin{align*}
&R\leq{1-h_2(p)}-o(1)\ (actual \ Shannon \  bound)\\
&R\geq{1-h_2(2p)}-o(1)\ (actual \ GV bound),
\end{align*}
this is mainly because of the exact bound we used for $p=\epsilon_{ss}\in{[(2k)^{-1},1/4]}$ (for both average and worst-case scenarios) rather than limiting $p\in{(0,1/4)}$. In light of this, Shannon entropy is necessary and sufficient condition to show meaningful security for a standard secure sketch.

\section{Comparison}\label{13.0}

\begin{center}
\begin{savenotes}
\begin{tabular}{|>{\raggedright\arraybackslash}p{3cm}| |>{\raggedright\arraybackslash}p{4cm}||>{\centering\arraybackslash}p{4.5cm}|  }
 \hline
 \multicolumn{3}{|c|}{Security Bound for Secure Sketch} \\
  \hline
{Computational}{{\ \ \ \ \ /Info-theoretic}} &Best possible security derived in \cite{fuller2016fuzzy} & $\minentropy{W}\geq{m}\geq\text{H}_{t,\infty}^
\text{fuzz}{(W)}-\log(1-\beta)$\\
 \hline
 \multirow{1}{*}{Computational}{{\ \ \ \ \ /Info-theoretic}} &FRS sketch(universal hash functions) ~\cite{fuller2016fuzzy} & {$\minentropy{W}\geq{m}\geq\text{H}_{t,\infty}^
\text{fuzz}{(W)}-\log(1/{\beta})-\log\log(\mathsf{supp}(W))-1$ }     \\   
 \hline
 \multirow{1}{*}{Computational}{\ \ \ \ \ /Info-theoretic} &Layer hiding hash  (strong universal hash function)\cite{woodage2017new} &   {$\minentropy{W}\geq{m}\geq\text{H}_{t,\infty}^
\text{fuzz}{(W)}-\log({1}/{\beta})-1$ }   \\
 \hline
\multirow{2}{*}{Info. theoretic} & {Fuzzy commitment with generic syndrome decoding \cite{juels1999fuzzy}} & {$\minentropy{W}\geq{t\log(n)}$ \ \ \ \ \ \ \ (when $t\ll{n}$) } \\ 
 \hline
\multirow{1}{*}{Info. theoretic} &{Fuzzy vault \cite{juels2006fuzzy}} &   {$\minentropy{W}\geq{t\log(n)}$ }   \\
 \hline
 \multirow{1}{*}{Info. theoretic} &{Improved Fuzzy vault \cite{dodis2004fuzzy}} &   {$\minentropy{W}>{t\log(n)+2}$ }   \\
 \hline
  \multirow{1}{*}{Info. theoretic} &{Pinsketch \cite{dodis2004fuzzy}} &   {$\minentropy{W}\geq{t\log(n+1)}$ }   \\
 \hline
  \multirow{1}{*}{Computational}{\ \ \ \ \ /info-theoretic} &\textbf{Proposed} & $\minentropy{W}\geq{m}\geq{{{{k^*h_2{(\epsilon_{rec})}}}}}=H(\mathcal{E}_{rec})$    \\
 \hline
\end{tabular}
\end{savenotes}
\end{center}
\begin{center}
Table 1: Summary of security bound of existing secure sketch in terms of fuzzy-min entropy (with $\beta>0$), min-entropy, and Shannon entropy. 
\end{center}

Table 1 depicted the security bound (upper bound min-entropy requirement) for existing secure sketch construction. Compared to the existing secure sketch constructions our construction is capable of \textit{claiming computational security for all sources with min-entropy $m\geq{1}$}.

An important highlight, we showed reduction from fuzzy-min entropy to Shannon entropy, hence our security can be simplified into merely depends upon the input error distribution defined using Shannon entropy, i.e., $m\geq{H(\mathcal{E}_{rec})}$. More precisely, by definition of min-entropy, $m$ never larger than $H(\mathcal{E}_{rec})$, thus $m=H(\mathcal{E}_{rec})$ to be exact, which means Shannon entropy is necessary and sufficient to describe $m$. Besides, it can be verified easily our derived security bound covered the best possible security stated in \cite{fuller2016fuzzy} with fuzzy min-entropy $\text{H}_{t,\infty}^
\text{fuzz}{(W)}\geq{1}$ (for $m\geq{1}$). 

\section{Conclusion}

Existing secure sketch constructions have shown limitation in providing security for sources with low entropy, i.e., lower than half of its input size. To overcome such limitation, recent approaches \cite{fuller2013computational}, \cite{woodage2017new} suggested to construct a secure sketch where its security property only holds for computationally bounded attacker. Such computational construction relies on fuzzy min-entropy measurement, accompanies with stringent requirement s.t. the user must have precise knowledge over the sources distribution. However, under the practical scenario, many sources, for instance biometric (human face, iris, and fingerprint) are difficult to model, hence assuming precise knowledge over such sources is unrealistic. 

In this work, we proposed an explicit construction for secure sketch. We adopted the usage of RV for sketching to facilitate the understanding of the input distribution. Besides, the noisy environment of any source with minimum entropy $m\geq{1}$ is stimulated by parsing random error to it. Our construction supports efficient recovery with bounded probability of error parametrized by the number of zero padding to the syndrome vector $v_{syn}$. The recovery algorithm works adversarially and efficiently under an event where the distribution of the random error $e\in\mathcal{E}_{ss}$ parsed into the input source is unknown. 

For security, we showed reduction from the usage of fuzzy min-entropy to Shannon entropy to claim security over any sources of min-entropy $m\geq{H(\mathcal{E}_{rec})}\geq{1}$. Our results allow the computational derivation of info-theoretic bound (Shannon bound), lead us to the calm over computational secure sketch implies info-theoretical secure sketch. 

\bibliographystyle{IEEEtran}
\bibliography{SS_HH_bib}

\end{document}